\documentclass[pdflatex,sn-mathphys-num]{sn-jnl}

\usepackage{amsmath,amssymb,amsthm,mathtools,bm}
\usepackage{graphicx}
\graphicspath{{./}{./figures/}}
\usepackage{booktabs}
\usepackage[capitalize,noabbrev]{cleveref}
\usepackage{enumitem}

\newtheorem{theorem}{Theorem}
\newtheorem{proposition}[theorem]{Proposition}
\newtheorem{corollary}[theorem]{Corollary}
\theoremstyle{definition}
\newtheorem{definition}[theorem]{Definition}
\theoremstyle{remark}
\newtheorem{remark}[theorem]{Remark}

\newcommand{\calK}{\mathcal{K}}
\newcommand{\calT}{\mathcal{T}}
\newcommand{\calC}{\mathcal{C}}
\newcommand{\EE}{\mathbb{E}}
\newcommand{\bagg}{\beta_{\mathrm{agg}}}
\newcommand{\bflow}{\beta_{\mathrm{flow}}}

\title[Privacy-Preserving Intent Fulfilment for 6G RAN]{Privacy-Preserving Intent Fulfilment and Assurance for 6G RAN}

\author*[1]{\fnm{Joss} \sur{Armstrong}}\email{joss.armstrong@ericsson.com}
\affil[1]{\orgname{Ericsson}, \orgaddress{\city{Athlone}, \country{Ireland}}}

\abstract{Intent-based network management is the emerging paradigm for 6G service
lifecycle automation, with the 3GPP intent framework
(TS~28.312) defining creation, translation, fulfilment, and assurance
stages. Existing fulfilment and assurance approaches require deep packet
inspection, per-flow state tracking, or access to vendor-internal node
telemetry to verify that provisioned resources satisfy expressed intents.
These requirements conflict with regulatory constraints (GDPR, ePrivacy
Directive) in multi-tenant networks and with vendor opacity in
O-RAN deployments.

We present an architecture for privacy-preserving intent fulfilment and
assurance in which a coordinator provisions resources from declared
intent categories without traffic inspection, and verifies fulfilment
using only cell-level standardised PM counters at the O1 interface. A
data-processing inequality argument shows that the resource allocation
reveals at most $\log_2 K$ bits about traffic content, where $K$ is the
number of intent categories. We define two architectural privacy
properties, intent-traffic unlinkability and node-opaque verification,
and show that both hold by construction. If verification needs at least
$q$ cells per category, the ceiling becomes
$\log_2 \lfloor M/q \rfloor$ bits for $M$ cells, so the detection
requirement fixes the leakage ceiling. When individual records add
state-independent noise to the category aggregate, the category-mean
verifier is at least as powerful as any per-agent verifier.

We map the architecture to the 3GPP intent lifecycle and the O-RAN
Non-RT RIC and identify the interfaces, data objects, and
deployment points where the mechanism operates. On production PM
data from four operator networks, the category-mean score has higher
AUC on three networks and comparable AUC on the fourth under
category-wide replacement,
while per-cell screening suits sparse faults. Co-movement of cells
within a reporting period limits a category mean to the evidence of 6
to 14 independent cells, and category-level verifiers need per-category
calibration.}


\begin{document}

\maketitle

% =======================================================================
\section{Introduction}
\label{sec:intro}
% =======================================================================

Intent-based networking (IBN) replaces manual, element-level
configuration with declarative service objectives. An operator expresses
an intent, for example ``ensure 99.9\% availability for enterprise
slice~X,'' and the management system translates, provisions, and assures
the corresponding resources autonomously. The 3GPP Service Assurance
(SA5) intent management framework, specified in TS~28.312
\cite{3gpp_28312}, defines a four-stage lifecycle of intent creation,
translation to network-level objectives, fulfilment through resource
provisioning, and continuous assurance through monitoring. The O-RAN
Alliance has adopted a compatible model in which a Non-RT RIC
orchestrates policy-driven management over the O1 and A1 interfaces
\cite{oran_wg1}.

Current approaches to intent fulfilment assume full observability of the
network state. The GenAI-IDN framework \cite{habib2025}
processes per-flow and per-session telemetry through large language models
to infer whether intents are satisfied. Survey literature on IBN for 6G
\cite{wang2026} identifies traffic classification, deep packet
inspection, and per-session monitoring as standard components of the
fulfilment and assurance pipeline. These assumptions create three
practical conflicts in production deployment.

First, deep packet inspection and per-flow monitoring conflict with data
protection regulations. The GDPR \cite{gdpr2016} requires data
minimisation, and the ePrivacy Directive \cite{eprivacy2002} restricts
interception of electronic communications content. An intent fulfilment
system that inspects traffic payloads or correlates per-user sessions
operates in tension with these requirements, even when the inspection
serves a legitimate network management purpose.

Second, multi-vendor O-RAN deployments create vendor opacity constraints.
The O-RAN architecture separates the RAN into components from different
vendors (CU, DU, RU), connected through open interfaces. Not all vendors
expose the same near-RT RIC or E2 telemetry. An intent assurance
mechanism that requires vendor-internal node state or proprietary
telemetry cannot operate consistently across a heterogeneous deployment.

Third, encrypted traffic, now exceeding 90\% of internet traffic
\cite{sharma2025}, renders content-based fulfilment increasingly
impractical. The encrypted traffic classification literature
\cite{sirinam2018,wang2014} treats this as a classification problem,
developing increasingly sophisticated inference attacks on metadata. An
architecture that never observes traffic content does not depend on
winning this arms race for its own decisions. Side channels outside the
architecture remain, and \cref{sec:limitations} discusses them.

The MISES mechanism \cite{mises_preprint} addresses these conflicts by
restricting the coordinator's information flow. The coordinator observes
only a finite category signal $C \in \{1, \ldots, K\}$ for provisioning
and only cell-level PM counters for assurance. No per-flow payloads,
per-agent traces, or vendor-internal state are accessed. The foundational
analysis \cite{mises_preprint} establishes the formal properties of this
restriction. The welfare loss of category-level provisioning is bounded
above and below by the within-category variation of preferred
allocations, and is proportional to it under common quadratic utility
(Theorem~1 there). When an individual record adds state-independent
noise to the category aggregate, the aggregate is a sufficient
experiment, so no individual-record test is more powerful (Theorem~2
there). Along nested category refinements, the designs that meet a
welfare tolerance and a per-category power target form an exact interval
(Theorem~3 there). On an explicit finite population that interval is
exact over every partition and profile choice (Theorem~4 there). This
paper develops the privacy consequences of that architecture and maps
them to the concrete deployment context of 6G intent-based management.

The contributions are:
\begin{enumerate}[leftmargin=*,nosep]
\item A leakage bound showing that the intent fulfilment mechanism
  reveals at most $H(C) \leq \log_2 K$ bits about traffic content,
  derived from the data processing inequality on the architectural
  Markov chain (\cref{sec:leakage}).
\item Two architectural privacy definitions, intent-traffic
  unlinkability and node-opaque verification, that hold by construction
  (\cref{sec:privacy}).
\item A structural analysis showing that privacy and detection are
  aligned. The minimum number of cells $q$ that each category needs for
  verification caps the category entropy, so that
  $I(T;\, R) \leq H(C) \leq \log_2 \lfloor M/q \rfloor$ for $M$ cells,
  and the detection requirement fixes the leakage ceiling
  (\cref{sec:tension}).
\item A deployment mapping to the 3GPP intent lifecycle and the O-RAN
  Non-RT RIC that identifies the interfaces, data objects, and deployment
  points at which the architecture operates (\cref{sec:deployment}).
\item An empirical study on production PM data from four operator
  networks (\cref{sec:empirical}). It compares per-cell and
  category-mean verification on a common decision unit, measures the
  mechanisms behind their network-dependent ordering, and tests
  counter-level degradation scenarios. Ancillary results
  include seed-variance bands on the detection scaling curve
  (\cref{sec:empirical:plateau}),
  a pure-DP randomized-response utility comparison at the released
  category-alarm interface
  (\cref{sec:empirical:strict}), and an empirical tightening of the
  $\log_2 K$ leakage bound via the learned partition entropies
  (\cref{sec:limitations}).
\end{enumerate}

% =======================================================================
\section{Background and Related Work}
\label{sec:related}
% =======================================================================

% -----------------------------------------------------------------------
\subsection{Intent-Based Network Management}
\label{sec:related:ibn}
% -----------------------------------------------------------------------

Intent-based networking emerged from the need to abstract network
management away from element-level configuration. The 3GPP intent
management framework (TS~28.312 \cite{3gpp_28312}, TR~28.912
\cite{3gpp_28912}) defines intents as declarative service objectives with
associated fulfilment and assurance stages. An intent object specifies
expectations (e.g.\ coverage area, throughput targets, availability)
and the management system autonomously translates these into network
configuration and continuously monitors fulfilment. The TM~Forum
autonomous networks initiative (IG1305 \cite{tmforum_ig1305}) and the
ETSI Zero-touch Service Management framework (ZSM \cite{etsi_zsm})
pursue compatible objectives.

Wang et al.\ \cite{wang2026} survey intent-driven end-to-end 6G system
designs and identify per-flow telemetry processing and traffic
classification as standard components. Leivadeas and Falkner
\cite{leivadeas2023} survey intent-based networking abstractions and note
the gap between declarative intent languages and the monitoring
capabilities needed for assurance. Habib et al.\ \cite{habib2025}
propose a GenAI-IDN (generative-AI intent-driven networking) architecture for 6G~RAN that processes
per-flow session data through a Mamba-based model for intent translation
and assurance. All of these approaches assume the management system has
access to granular per-flow or per-session telemetry. Software-defined
and knowledge-defined frameworks for 5G/6G multimedia Internet of
Vehicles make the same assumption, with a central controller that uses
flow-level state to optimise vehicle routing, packet routing, and
resource allocation jointly \cite{montazerolghaem2026r21} or to balance
load and energy \cite{montazerolghaem2026pid}, whereas the architecture
here asks what such latency- and energy-sensitive control can retain
when the controller sees only declared categories and cell-level
counters.

% -----------------------------------------------------------------------
\subsection{Privacy in Network Management}
\label{sec:related:privacy}
% -----------------------------------------------------------------------

Differential privacy \cite{dwork2006,abadi2016} provides formal
guarantees for statistical queries over sensitive data by adding
calibrated noise. Federated learning \cite{mcmahan2017} keeps raw data
at local nodes while training shared models, though recent work shows
gradient leakage attacks can partially invert this protection
\cite{blika2024}. Secure multi-party computation \cite{yao1986} and
homomorphic encryption \cite{gentry2009} enable computation on encrypted
inputs. These approaches share a common structure. A full-information
pipeline exists, and a protection mechanism is applied after or during
data collection. The MISES architecture differs by restricting the
information flow so that the data requiring protection is never collected.

Traffic analysis and website fingerprinting \cite{sirinam2018,wang2014}
study what an adversary can infer from encrypted traffic metadata.
Unlinkability, the inability to link observed actions to identities, is a
standard privacy goal formalised by Pfitzmann and Hansen
\cite{pfitzmann2010}. In the MISES architecture, the category signal is
the only link between the agent and the coordinator's resource decision.
The resulting unlinkability bound is architectural rather than
computational, and it holds regardless of the adversary's processing
power.

Privacy by design \cite{cavoukian2009} advocates building privacy into
system architecture rather than adding it as a compliance layer. The GDPR
data minimisation principle (Article~5(1)(c) \cite{gdpr2016}) requires
that only data necessary for a specified purpose be collected. The
architecture presented here instantiates both principles. The mechanism's
information flow is designed so that privacy-sensitive data is never
available to the coordinator.

% -----------------------------------------------------------------------
\subsection{Aggregate Monitoring and O-RAN}
\label{sec:related:oran}
% -----------------------------------------------------------------------

The 3GPP PM framework (TS~28.552 \cite{3gpp_28552}) defines standardised
performance measurement counters for 5G NR, including cell-level
throughput, PRB utilisation, HARQ statistics, latency metrics, and
active-UE counts. These counters are exposed at the O-RAN O1 interface
as bulk PM data files \cite{oran_wg1}. The O-RAN architecture separates
the RAN into open, multi-vendor components and defines the Non-RT RIC as
the management entity responsible for policy-based optimisation at
timescales above one second.

Polese et al.\ \cite{polese2023} provide a comprehensive survey of the
O-RAN architecture and note that the O1 interface provides cell-level
aggregate performance data while the E2 interface provides near-real-time
per-UE data, but E2 exposure varies significantly across vendors. The
architecture presented here operates exclusively on O1 cell-level data,
making it deployable in any O-RAN configuration regardless of E2
availability.

% =======================================================================
\section{System Model}
\label{sec:model}
% =======================================================================

A coordinator $\calC$ manages a shared resource pool for a population
of $M$ agents. Each agent has a demand type $t \in \calT$ and generates
traffic with content variable $T$. The coordinator does not observe $t$
or $T$.

Each agent declares a category $c \in \calK = \{1, \ldots, K\}$. In the
intent management context, categories correspond to intent classes,
groups of intents with similar resource requirements derived from the
intent translation stage. The coordinator provisions resources
$R = \phi(c)$ and verifies service fulfilment using cell-level PM
counters $\bm{X}_1, \ldots, \bm{X}_n$ produced within each category. No
per-flow payloads, per-agent traces, or internal node state are
accessed.

\subsection{Reporting and Aggregation Semantics}
Two distinct notions of ``aggregate'' recur below and are named
throughout. A \emph{cell-level aggregate} is a standardised O1 PM
counter for a single cell over one reporting period, pooled over the
users active in that cell during the period. This is what the verifier
consumes. Cell-level aggregates $\bm{X}_i$ are indexed by (cell,
period) and constitute the standard O1 interface output. A
\emph{category-level aggregate} is a further reduction across cells
within one intent category and one period, typically the sample mean
$\bar{\bm{X}}_k = (1/m_k) \sum_{i \in k} \bm{X}_i$. The distinction
matters for both privacy and detection. The coordinator's interface
exposes cell-level aggregates by design (matching the O1 baseline),
while the category-level aggregate is a downstream statistic the
verifier can compute from them. Where the text refers to $M$ agents,
those agents are the privacy subjects (users, sessions, flows). Cells
are the reporting units at which PM counters are produced. The
empirical count $M$ used in the illustration (\cref{sec:empirical}) is
the number of cells, since the four operator datasets expose no
per-user attribution and cell-level aggregation matches both the O1
interface and the standing GDPR data-minimisation regime.

The architecture enforces the Markov chain
\begin{equation}\label{eq:markov}
  T \;\longrightarrow\; C \;\longrightarrow\; R = \phi(C).
\end{equation}
Traffic content $T$ influences the category declaration, but the
coordinator observes only $C$ and acts only on $\phi(C)$. This Markov
structure is a design constraint enforced by the architecture, not an
analytical convenience.

\subsection{Welfare Model}

Let $W^*$ denote the welfare under a full-information benchmark where
the coordinator observes each agent's true demand type and provisions
accordingly. The category-based mechanism achieves welfare $W(K)$ with
gap $\Delta(K) = W^* - W(K) \geq 0$. Let $\varepsilon(K)$ be the mean
squared distance between each agent's preferred allocation and the
profile of its category. Theorem~1 of \cite{mises_preprint} bounds the
gap on both sides. If the curvature of each agent's utility lies in
$[\mu, L]$, then
$\tfrac{\mu}{2}\varepsilon(K) \leq \Delta(K) \leq \tfrac{L}{2}\varepsilon(K)$
for any menu and any category rule. Under common quadratic utility the
gap is proportional to $\varepsilon(K)$. When each category receives its
mean profile, refining a partition cannot increase $\varepsilon$
(Corollary~2 there). Finer nested categories therefore yield better
provisioning.

\subsection{Detection Model}

The coordinator tests service fulfilment for each category using
cell-level PM counters. Under the null hypothesis (fulfilled service),
$\bm{X}_i \sim P_k$ with mean $\bm{\mu}_k$. Under the alternative
(degraded service), $\EE[\bm{X}_i] = \bm{\mu}_k - \bm{\delta}_k$,
where $\bm{\delta}_k$ is a fixed degradation vector. The category-level
sample mean $\bar{\bm{X}}_k = (1/m_k)\sum_{i \in k} \bm{X}_i$ is one
natural test statistic. A per-cell distance to the trained category
centroid is another. Both are consumed by the same interface (a
downstream reduction of cell-level aggregates within a category), and
the choice between them is a detection-side design decision. The
empirical instantiation in \cref{sec:empirical} uses per-cell
distances, and \cref{sec:empirical:strict} additionally measures the
category-mean variant.

We write $\bagg(\alpha)$ for the best attainable detection power of a
test on the category aggregate at false-alarm level $\alpha$, and
$\bflow(\alpha)$ for the best power of a test on individual agent
records $\bar{y}_i$.

The formal results in \cite{mises_preprint} rest on explicit
assumptions. The welfare results need bounded utility curvature
(Assumptions~1--2 there). The verification results need
\emph{state-independent observation} (Assumption~3 there). An observer
holding the category aggregate must be able to generate a record
distributed like the individual record under both the fulfilled and the
degraded state. The additive model
$\bar{y}_i = \bar{\bm{X}}_k + \bar{\xi}_i$, in which each individual
record is the category aggregate plus noise $\bar{\xi}_i$ independent of
$\bar{\bm{X}}_k$ and of the service state, is one instance. Neither assumption requires Gaussian noise or equal category
sizes. The leakage results of \cref{sec:privacy} need only the Markov
chain~\eqref{eq:markov}.

\subsection{Threat Model}

We consider an honest-but-curious coordinator that follows the protocol
but may attempt to infer information about $T$ from $R$ and the
cell-level PM counters. A malicious coordinator that deviates from the
protocol, for instance by deploying DPI outside the mechanism, is outside
the threat model. The definitions and bounds below characterise what the
mechanism reveals when followed as designed. \Cref{sec:limitations}
gives a fuller formal statement of the threat model and the excluded
adversary classes.

% =======================================================================
\section{Privacy Architecture}
\label{sec:privacy}
% =======================================================================

The MISES mechanism maps declared intent categories to resource
allocations without inspecting traffic content. This induces a structural
constraint on what the allocation can reveal.

% -----------------------------------------------------------------------
\subsection{Leakage Bound}
\label{sec:leakage}
% -----------------------------------------------------------------------

\begin{theorem}[Category-Limited Leakage Bound]\label{thm:dpi}
Assume the MISES architecture induces the Markov chain
$T \to C \to R$ as in~\eqref{eq:markov}. Then
\begin{equation}\label{eq:dpi}
  I(T;\, R) \leq I(T;\, C) \leq H(C).
\end{equation}
If the category prior is uniform over $K$ categories,
$I(T;\, R) \leq \log_2 K$.
\end{theorem}

\begin{proof}
The first inequality is the data-processing inequality
\cite{cover2006}. Since $T \to C \to R$ is Markov, no processing of $C$
can increase the mutual information between $T$ and the output, which
gives $I(T;\, R) \leq I(T;\, C)$. The second inequality is the entropy bound
on mutual information,
$I(T;\, C) = H(C) - H(C \mid T) \leq H(C)$. For a uniform prior over
$K$ categories, $H(C) = \log_2 K$.
\end{proof}

The ceiling is architectural. It is determined by the number of
categories and their prior distribution, not by any additional privacy
mechanism. Reducing semantic recoverability at fixed $K$, for example by
using demand-derived rather than semantic partitions, can reduce
$I(T;\, C)$ below $H(C)$ and tighten the effective leakage. But $H(C)$
is the hard ceiling regardless of partition design.

The bound does not claim that MISES achieves minimum leakage among all
possible mechanisms for the coordination task. A mechanism that adds
noise to $R$ or uses fewer than $K$ effective categories could achieve a
tighter bound. The result is that MISES leaks at most $\log_2 K$ bits
about traffic content, not that $\log_2 K$ is the best achievable.

\begin{remark}[Comparison with Differential Privacy]
Differential privacy provides an $(\varepsilon, \delta)$-guarantee on
the output distribution's sensitivity to any single record. The leakage
bound in \cref{thm:dpi} provides a mutual-information ceiling on the
total information the output reveals about the input. These are
complementary notions. DP bounds worst-case distinguishability. The
data-processing bound limits aggregate information flow. Neither subsumes
the other. The practical distinction is that the data-processing bound
holds by architectural
construction without noise injection, while DP requires a randomised
mechanism that may degrade utility.
\end{remark}

% -----------------------------------------------------------------------
\subsection{Privacy Definitions}
\label{sec:privdefs}
% -----------------------------------------------------------------------

\begin{definition}[Intent-Traffic Unlinkability]\label{def:unlinkability}
Let $T$ denote the traffic-content variable, $C$ the category signal,
and $R = \phi(C)$ the provisioned resource decision. The architecture
is \emph{category-limited unlinkable} if
\begin{equation}\label{eq:unlinkability}
  I(T;\, R) \leq I(T;\, C) \leq H(C).
\end{equation}
\end{definition}

Category-limited unlinkability states that the resource allocation
reveals no more about traffic content than is already present in the
category signal, and the category signal itself carries at most $H(C)$
bits. When the category prior is uniform over $K$ categories, the
ceiling reduces to $\log_2 K$ bits. Reducing $K$ tightens the ceiling
but at the cost of provisioning accuracy (Theorem~1 in \cite{mises_preprint}).

In the intent management context, $T$ captures the actual service
behaviour (application mix, traffic patterns, user activity), $C$ is
the intent category declared during the translation stage, and $R$ is
the resource configuration provisioned during fulfilment. Unlinkability
means that an observer of the provisioned configuration cannot
reconstruct the underlying traffic behaviour beyond what the intent
category already reveals.

\begin{definition}[Node-Opaque Verification]\label{def:node_opacity}
A verification mechanism is \emph{node-opaque} if its decision statistic
is a measurable function only of standardised O1 cell-level PM counters,
and does not require access to:
\begin{enumerate}[nosep,label=(\alph*)]
\item per-flow payloads or deep packet inspection,
\item per-UE and per-flow traces and metrics, or
\item vendor-internal node state or proprietary telemetry.
\end{enumerate}
\end{definition}

Node-opacity is a property of the verification interface, not a claim
about what data exists elsewhere in the system. The coordinator may have
access to per-agent data through other channels. The definition
constrains the verification mechanism so that its decision does not
depend on that data.

In the O-RAN context, node-opacity means the assurance function operates
entirely on O1 bulk PM data, without requiring E2-level per-UE telemetry
or vendor-specific near-RT RIC data. This makes the assurance mechanism
deployable across heterogeneous multi-vendor O-RAN deployments where E2
exposure is inconsistent.

\begin{proposition}[Category-Mean vs Per-Agent Verification]
\label{prop:node_opacity}
Fix the category, the monitoring period, and the fulfilled and degraded
observation laws. Under state-independent observation (Assumption~3 of
\cite{mises_preprint}), for every false-alarm level
$\alpha \in [0, 1]$,
\begin{equation}\label{eq:node_opacity}
  \bagg(\alpha) \geq \bflow(\alpha).
\end{equation}
An observer holding both the aggregate and the individual records has
the same optimal power as one holding the aggregate alone.
\end{proposition}

\begin{proof}
This is Theorem~2 of \cite{mises_preprint}, a comparison-of-experiments
argument \cite{blackwell1953}. Take any test on the individual record.
An observer with the aggregate can generate a simulated individual
record with the same law under both states and run that test, with
unchanged size and power. The same simulation reproduces any test that
uses both records.
\end{proof}

The result is distribution-free. In the Gaussian additive model the
dominance is strict for finite $m$ and positive individual noise, and
the power gap is at most of order $1/m$ at fixed noise and shift
(Appendix~A of \cite{mises_preprint}).

\begin{remark}[When the premise fails]\label{rem:tworegimes}
Proposition~\ref{prop:node_opacity} needs each individual record to be
reproducible from the category aggregate under both service states. A
common category-wide shift with state-independent residual noise
satisfies this. Sparse degradation can violate the premise. Suppose a fraction $w$ of
the $m$ cells in a category move, by displacements $\bm{\Delta}_i$. The
category mean moves by $(1/m)\sum_i \bm{\Delta}_i$, whose norm is at
most $w$ times the mean $\|\bm{\Delta}_i\|$, while a per-cell test sees
each $\|\bm{\Delta}_i\|$ in full. Two further effects shrink the
category-mean signal on production data. Displacements in different
directions partly cancel in the sum. Cells that move together within a
reporting period inflate the variance of the category mean above
$\sigma^2/m$, so $m$ cells act as fewer effective samples. Under sparse
degradation neither verifier dominates in general.
\Cref{sec:empirical:strict} measures each effect on the four production
networks.
\end{remark}

Node-opacity is therefore a security property, not merely a
data-collection preference. The verification interface consumes only
standard O1 cell-level PM counters and adds no exposure beyond the
O-RAN O1 baseline. The strict category-mean variant reduces the
interface exposure further, to $K$ scalar statistics per period
(one category mean per intent category). \Cref{sec:empirical:strict}
measures when that reduction costs detection power and when it does
not.

% -----------------------------------------------------------------------
\subsection{Structural Alignment of Privacy and Detection}
\label{sec:tension}
% -----------------------------------------------------------------------

The category signal $C$ serves two epistemically distinct purposes,
provisioning (allocating resources to match demand) and verification
(detecting whether the allocation succeeded). These two purposes impose
structurally opposed demands on signal granularity.

Provisioning favours finer categories. Along a nested refinement with
mean profiles, the within-category variation $\varepsilon$ cannot
increase, so under common quadratic utility the welfare gap cannot
increase either (Corollary~2 and Lemma~1 of \cite{mises_preprint}).
Verification favours coarser categories. Every category must be
verified on its own, so the power that matters is that of the least
detectable category. Along a nested refinement in which each parent's
aggregate can reproduce each child's, at a false-alarm allowance that
does not grow, that least power cannot increase (Lemma~2 there). This
ordering concerns a test on the category aggregate. A per-cell screener
that measures distance to a category centroid can gain from tighter
centroids as $K$ grows, which is what the headline screener of
\cref{sec:empirical:plateau} does, and that gain does not contradict
the ordering.
Theorem~3 of \cite{mises_preprint} turns the two orderings into an exact
statement. Fix a welfare tolerance $\varepsilon^*$ and a power target
$\beta^*$. The levels of the nested family that meet both targets form
one contiguous interval $[K_{\min}, K_{\max}]$, the feasibility band. It
is empty exactly when the coarsest level that meets the welfare
tolerance already fails verification. On an explicit finite population,
Theorem~4 there shows that the band is exact over every partition and
every profile choice, with a matching converse.

The verification side also has a population form. Suppose each category
needs at least $q$ contributing cells to reach the power target at the
stated false-alarm allowance. Then a design over $M$ cells in which
every category is verifiable has $K \leq \lfloor M/q \rfloor$
(Corollary~3 of \cite{mises_preprint}, the occupancy ceiling). The
leakage bound (\cref{thm:dpi}) adds a third axis to this tradeoff.
Since $I(T;\, R) \leq H(C) \leq \log_2 K$, privacy and detection are
aligned. Both favour fewer, larger categories, and welfare opposes both.

\begin{corollary}[Detection-Capped Leakage]\label{cor:capped}
If every category contains at least $q$ cells, the number each needs to
meet the power target, then
\begin{equation}\label{eq:capped}
  I(T;\, R) \;\leq\; H(C) \;\leq\; \log_2 \lfloor M/q \rfloor .
\end{equation}
\end{corollary}

\begin{proof}
\Cref{thm:dpi} gives $I(T;\, R) \leq H(C) \leq \log_2 K$. A partition
of $M$ cells into $K$ categories of at least $q$ cells each has
$Kq \leq M$.
\end{proof}

The verification requirement therefore fixes the privacy ceiling. No
separate privacy parameter is needed for the category channel. The
ceiling also has a matching floor. Corollary~4 of \cite{mises_preprint}
shows that any design meeting the welfare tolerance has
$H(C) \geq R_P(\varepsilon^*)$, where $R_P$ is an explicit floor that
depends only on the demand source and the tolerance. Every design that
meets both targets satisfies
\begin{equation}\label{eq:envelope}
  R_P(\varepsilon^*) \;\leq\; H(C) \;\leq\; \log_2 \lfloor M/q \rfloor .
\end{equation}
Provisioning sets the least information the category must carry, and
verification sets the most it may carry. If the floor exceeds the
ceiling, no category design meets both targets. On the explicit
population of \cite{mises_preprint}, at balanced operating points,
$\log_2 K$ is the minimum category entropy over all designs that meet
both targets (Corollary~6 there). There the leakage ceiling is the
information that provisioning at that tolerance requires.

\paragraph{Merging and splitting under load.}
The ceiling in~\eqref{eq:capped} moves with the population. When load
falls, fewer cells contribute, the occupancy ceiling
$\lfloor M_t/q \rfloor$ in period $t$ falls, and categories below $q$
must merge. For a fixed, public merge map, the new label is a
deterministic function of the existing label, so it is post-processing
of $C$ and cannot raise $H(C)$. When load rises,
categories may split, but only up to $\lfloor M_t/q \rfloor$. A split
asks each agent which child it belongs to, which is a new declaration.
Across periods the bound therefore composes by the chain rule. For
declarations $C_1, \ldots, C_n$,
$I(T;\, C_1, \ldots, C_n) \leq \sum_t H(C_t) \leq
\sum_t \log_2 \lfloor M_t/q \rfloor$. A single lifetime ceiling
applies when the intent retains its original category value and later
disclosures only repeat or post-process that value through fixed, public
maps. An unchanged partition alone does not ensure this when demand
changes and agents declare new category values.
\Cref{sec:empirical:scaling} reports the effect of enforcing $q$ on the
four production networks.

These relations hold under the stated assumptions and do not depend on
the semantic content of the categories. Where the operator sits inside
the band cannot be derived from the formal structure alone. The
provisioner wants high $K$ for tight categories and low $\varepsilon$.
The verifier wants low $K$ for large per-category populations, high
detection power, and low leakage. The architecture characterises this
conflict rather than resolving it. The choice reflects the operator's
priorities, not an optimisation problem with a unique solution.

% =======================================================================
\section{Deployment in the 3GPP Intent Lifecycle}
\label{sec:deployment}
% =======================================================================

This section maps the MISES architecture to the concrete interfaces and
data objects of the 3GPP intent management framework and the O-RAN
architecture. The mapping identifies where each component of the
mechanism operates in a production 6G deployment.

% -----------------------------------------------------------------------
\subsection{Intent Lifecycle Mapping}
\label{sec:deployment:lifecycle}
% -----------------------------------------------------------------------

The 3GPP intent lifecycle (TS~28.312 \cite{3gpp_28312}) defines four
stages. The MISES mechanism maps to each as follows.

\textbf{Stage 1: Intent Creation.}
The operator creates an intent object specifying service expectations
(coverage, throughput, availability, latency targets). No modification
to this stage is required. The intent object is the input to the
management system.

\textbf{Stage 2: Intent Translation.}
The management system translates the intent into network-level
objectives. In the MISES architecture, translation includes category
assignment. Each intent (or group of intents with similar resource
profiles) is mapped to one of $K$ intent categories. The category
assignment is the demand-derived partition from \cite{mises_preprint}, constructed from the declared resource
requirements, not from traffic content. This is the point at which the
information restriction is established. The translator outputs a
category label $c \in \{1, \ldots, K\}$ and the corresponding resource
profile $\phi(c)$.

The category construction can use historical demand distributions
without per-flow traffic data. Standard clustering on declared intent
parameters (requested throughput, latency bound, availability target)
produces the partition. The number of categories $K$ is a system design
parameter. It must meet the welfare tolerance and keep every category at
or above the verifier's minimum occupancy $q$, so that
$K \leq \lfloor M/q \rfloor$ (\cref{sec:tension}).

\textbf{Stage 3: Intent Fulfilment.}
The management system provisions resources according to $\phi(c)$.
This is Mechanism~A (blind intent fulfilment). The coordinator allocates resources based solely
on the category label. No traffic inspection, session correlation, or
per-user state is required during provisioning. The fulfilment
function $\phi$ maps each category to a resource configuration (PRB
allocation profile, scheduling priority, QoS flow configuration) using
standard 3GPP provisioning interfaces.

\textbf{Stage 4: Intent Assurance.}
The management system continuously monitors whether the provisioned
resources satisfy the expressed intent. This is Mechanism~B (node-opaque
outcome verification). The assurance function observes only cell-level
PM counters from the O1 interface. For each intent category the
assurance function operates as a two-tier screener:
\begin{enumerate}[leftmargin=*,nosep]
\item \emph{Per-cell screening.} Each cell-level PM vector is compared
  against the category centroid trained on the validation slice. A
  cell is flagged in a period if its distance to the centroid exceeds
  a network-level threshold, set as a single quantile of the pooled
  validation cell-hour distance distribution at a target per-cell false
  positive rate $\alpha_0$. One threshold is stored per network and
  applied uniformly across categories.
\item \emph{Category-level alarm.} The per-category flagged-fraction
  in each period is monitored against a network-level exceedance
  threshold, calibrated as a single quantile of the pooled validation
  per-category flagged-fraction distribution at an operator-chosen
  category-level null rate. The alarm is the category-level signal on
  which remediation is triggered. Operating constraints of the tail
  parameterisation of this alarm are characterised in
  \cref{sec:empirical:alpha0}.
\end{enumerate}
The strict category-mean variant (\cref{sec:empirical:strict})
replaces per-cell screening by a single distance from the category
mean to the trained centroid. It is the better screener when degradation
affects a category as a whole and the weaker one for sparse faults.

The assurance loop is closed. If the category-level alarm fires, the
management system can trigger re-provisioning at category granularity
without needing to identify which specific intents or users are
affected.

% -----------------------------------------------------------------------
\subsection{O-RAN Deployment Architecture}
\label{sec:deployment:oran}
% -----------------------------------------------------------------------

In the O-RAN architecture, the MISES mechanism deploys as a Non-RT RIC
rApp (or equivalently as a 3GPP MnS Producer). The deployment uses two
standard interfaces:

\textbf{O1 Interface (PM Data).}
The O1 interface provides bulk PM data files containing standardised
counters as defined in TS~28.552 \cite{3gpp_28552}. These include:
\begin{itemize}[nosep]
\item Cell-level DL/UL throughput and PRB utilisation
\item HARQ ACK/NACK ratios and modulation order distributions
\item RRC connection statistics
\item Active UE counts per cell
\item DL/UL PDCP SDU volumes
\end{itemize}
The assurance function (Mechanism~B) operates exclusively on these
counters. No per-UE, per-bearer, or per-flow data crosses the O1
interface in the standard PM data model.

\textbf{A1 Interface (Policy).}
The A1 interface carries policy directives from the Non-RT RIC to the
near-RT RIC. In the MISES deployment, A1 policies encode the
category-to-resource mapping $\phi(c)$. For each intent category, the
policy specifies the target resource configuration. Updates to $\phi$
(e.g.\ when load changes move the occupancy ceiling
$\lfloor M/q \rfloor$) are
propagated as A1 policy updates.

\textbf{No E2 Dependency.}
The mechanism does not require E2 interface data (per-UE metrics,
per-bearer reports). This is the key deployment advantage. The E2
interface provides near-real-time granular telemetry, but its
availability varies across vendors. Some O-RAN vendors expose rich E2
data, and others expose minimal subsets or none at all. By operating
exclusively on O1 cell-level data, the mechanism is deployable in any
O-RAN configuration without vendor-specific integration.

% -----------------------------------------------------------------------
\subsection{Multi-Vendor and Multi-Tenant Considerations}
\label{sec:deployment:multiv}
% -----------------------------------------------------------------------

The node-opacity property (Definition~\ref{def:node_opacity}) has direct
operational consequences in multi-vendor deployments.

\textbf{Vendor heterogeneity.}
In a deployment with CU from vendor~A, DU from vendor~B, and RU from
vendor~C, the assurance function does not require any vendor to expose
internal processing state. The O1 PM counters defined in TS~28.552 are
standardised and vendor-agnostic. The assurance decision depends on
cell-level PM counters that all compliant vendors must report.

\textbf{Multi-tenant isolation.}
In network slicing scenarios, each tenant's intents map to separate
intent categories. The assurance function monitors each category's
cell-level counters independently. No cross-tenant per-flow correlation
is required. The leakage bound (\cref{thm:dpi}) applies per-category.
Tenant~A's category declaration and resource allocation reveal at most
$\log_2 K_A$ bits about tenant~A's traffic, where $K_A$ is the number
of intent categories allocated to tenant~A. This ceiling does not bound
additional inference from the cell-level PM counters used for assurance.

\textbf{Regulatory alignment.}
The architecture supports data minimisation (GDPR Article~5(1)(c)) by
construction, since the management system does not collect per-flow or
per-user data for intent fulfilment or assurance. For deployments subject
to the ePrivacy Directive, the assurance function observes
only cell-level performance counters, not user-plane data, and therefore
performs no interception of communications content. These are
architectural properties, not a compliance determination. Whether a
given deployment meets its GDPR and ePrivacy obligations depends on the
controller's overall processing and remains a deployment-specific legal
assessment.

% -----------------------------------------------------------------------
\subsection{Implementation Considerations}
\label{sec:deployment:impl}
% -----------------------------------------------------------------------

\textbf{Category construction.}
The intent categories are constructed offline from historical declared
intent parameters using standard clustering (e.g.\ $k$-means on the
requested-QoS vector). The partition is updated periodically (e.g.\
daily or weekly) as the intent population evolves. The number of
categories $K$ is chosen between the smallest value that meets the
welfare tolerance and the occupancy ceiling $\lfloor M/q \rfloor$,
based on the operator's welfare-detection preference.

\textbf{Threshold calibration.}
Two thresholds are calibrated at the network level from a validation
period of cell-level PM data under known-good conditions. The per-cell
distance-to-centroid threshold is a single quantile of the pooled
validation cell-hour distance distribution, set to the target per-cell
false positive rate $\alpha_0$. The same threshold is applied across
all categories. The category-level exceedance threshold is a single
quantile of the pooled validation per-category flagged-fraction
distribution, set to the target category-level alarm rate on the null.
Both calibrations use only cell-level data and do not require
labelled per-flow ground truth. A category-level statistic such as the
strict category mean must be calibrated per category. A single
quantile pooled across categories can leave it at chance
(\cref{sec:empirical:strict}). The internal $\alpha_0$ is a
configuration knob, not a deployment claim. The operator-facing rate
is the category-level alarm rate, and both are deployer-tunable
independently.

\textbf{Assurance loop timing.}
The assurance loop runs at the granularity of the O1 PM reporting
period, typically 15 minutes for bulk PM files or sub-minute when O1
streaming PM is available. The mechanism is agnostic to the reporting
rate as it requires only cell-level counters per period. By contrast,
DPI-based and per-flow approaches must process traffic continuously at
line rate, and differential privacy mechanisms add per-query noise
calibration overhead. The MISES assurance loop has no per-packet or
per-query processing cost.

\textbf{Graceful degradation.}
If the assurance function detects degradation in a category, the
response is a category-level action (re-provisioning, policy update
via A1) rather than a per-user action. This preserves node-opacity
through the remediation stage. The management system acts on categories,
not individual agents.

% =======================================================================
\section{Empirical Illustration}
\label{sec:empirical}
% =======================================================================

We evaluate the architecture on five weeks of PM counter data
(2025-W45 through 2025-W49) from four anonymised production operator
networks. \Cref{tab:networks} summarises the datasets.

\begin{table}[t]
\centering
\caption{Production network PM datasets.}
\label{tab:networks}
\begin{tabular}{lrrr}
\toprule
Network & Cells & Cell-hours & Features \\
\midrule
Net-A & 537 & 431{,}603 & 11 \\
Net-B & 4{,}987 & 4{,}017{,}762 & 11 \\
Net-C & 12{,}861 & 10{,}624{,}114 & 11 \\
Net-D & 6{,}136 & 5{,}061{,}561 & 11 \\
\bottomrule
\end{tabular}
\end{table}

Each cell contributes one cell-level PM vector per reporting period,
covering throughput, resource utilisation, HARQ modulation profile,
latency proxy, and active UE count (eleven standardised features across
four resource dimensions). No per-agent or per-flow data is used at
any stage. The PM counters correspond to the standardised TS~28.552
measurements described in \cref{sec:deployment:oran}.

\subsection{Instantiation and Ground Truth}
\label{sec:empirical:instantiation}

Intent categories are constructed by $k$-means clustering on per-cell
mean PM profiles from the training period (first 60\% of hours,
temporal split). The validation set (next 20\%) calibrates the
distance-to-centroid threshold at a target false positive rate
$\alpha_0 = 0.20$. The test set uses injection-based ground truth.
A fraction $\rho = 0.30$ of test-set cell-hours have their PM vectors
replaced by vectors drawn from a different cluster's training reservoir.
The ground truth label is set before the detection test runs and is not
derived from any distance computation.

This instantiates the deployment model from \cref{sec:deployment}.
Categories correspond to intent classes. The detection statistic is
the per-cell distance from a cell-hour PM vector to its assigned
category centroid. A cell is flagged in a period when this distance
exceeds a single network-level threshold, calibrated as the
$(1-\alpha_0)$ quantile of the pooled validation cell-hour distance
distribution. The paper's headline numbers use
this per-cell instantiation. \Cref{sec:empirical:strict} compares it
with the strict category-mean variant on a common (category, period)
decision unit. The threshold calibration on the held-out period mirrors the
deployment calibration described in \cref{sec:deployment:impl}. The
injection-based degradation represents a scenario where a subset of
cells in a category experience a service quality shift, which is the
event the assurance function must detect. The per-cell verifier is
the deployed screener at this operating point. Remark~\ref{rem:tworegimes}
explains how this regime can violate the premise of
Proposition~\ref{prop:node_opacity}, and \cref{sec:empirical:scenarios} replaces
the cross-category draw with counter-level perturbations.

\subsection{Detection Plateau}
\label{sec:empirical:plateau}

\Cref{fig:sweep} shows detection recall across
$K \in \{3, 5, 8, 10, 15, 20, 25, 30\}$ at fixed FPR~$= 0.20$.
Three regimes are visible.

\begin{figure}[t]
\centering
\includegraphics[width=0.78\columnwidth]{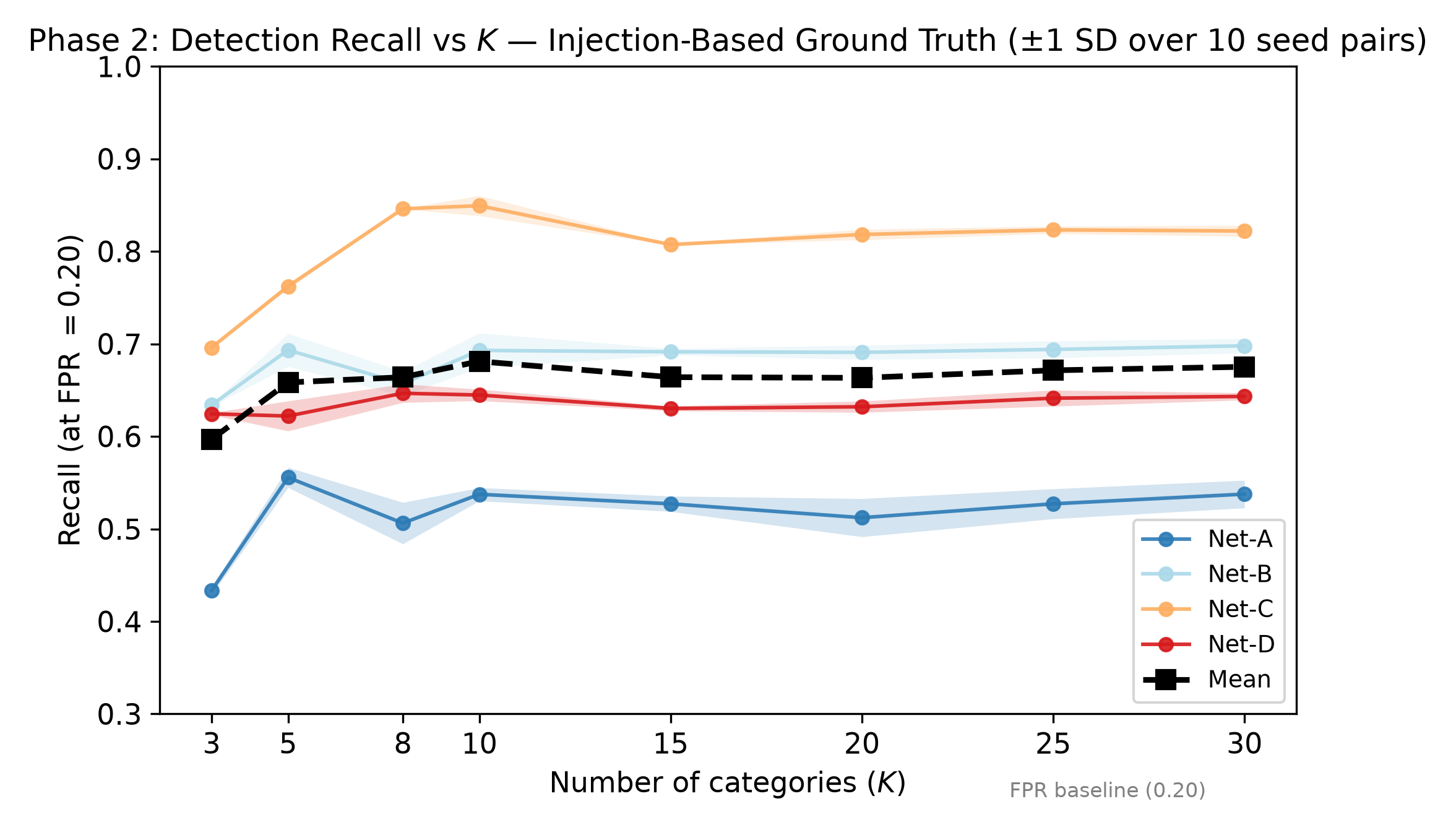}
\caption{Detection recall vs.\ $K$ at FPR~$= 0.20$ across four
         production networks, with $\pm 1$ SD bands over 10 paired
         seeds (KMeans random state, injection RNG seed). In all four
         networks recall plateaus while the leakage ceiling $\log_2 K$
         continues to increase. Generated from the same pipeline as
         \cref{tab:scaling,tab:staleness} for provenance parity.}
\label{fig:sweep}
\end{figure}

At low $K$ ($K = 3$), categories are too coarse for meaningful null
distributions and recall is weakest (0.433--0.696 across networks).
At moderate $K$ ($K = 5$--$8$), tighter categories improve
discrimination. Net-C reaches 0.85 recall at $K = 8$.
At high $K$ ($K \geq 10$), recall plateaus as the tightening gain is
offset by loss of inter-cluster separation.

The plateau is consistent with the dual-purpose tension predicted by the
theory.
In these networks, beyond $K \approx 8$ the signal is serving the
provisioner while the verifier gains nothing further.
The leakage ceiling $\log_2 K$ continues to grow with $K$, from 3.0 bits
at $K = 8$ to 4.9 bits at $K = 30$. The designer who increases $K$ past
the detection plateau is paying a privacy cost for a provisioning gain,
with no detection benefit.

\subsection{Detection Scaling with Network Size}
\label{sec:empirical:scaling}

The detection plateau appears consistently across all four networks
despite an order-of-magnitude range in network size (537 to 12{,}861
cells). Larger networks exhibit higher absolute recall. Net-C (12{,}861
cells, 1{,}286 cells per category at $K = 10$) reaches 0.85, while
Net-A (537 cells, 54 per category at $K = 10$) reaches 0.53. Category
size is the quantity the occupancy ceiling of \cref{sec:tension}
constrains, but it does not determine recall on its own.

\Cref{tab:scaling} shows how recall varies with the cells-per-category
ratio $M/K$ across networks and cluster counts. For the two networks
shown, Net-C (higher $M/K$) exceeds Net-A at every $K$. Across all
four production networks the
ordering is not strict in $M/K$. Net-B exceeds Net-D at
every $K$ in the sweep of \cref{fig:sweep} despite having fewer
cells per category.

\begin{table}[t]
\centering
\caption{Detection recall by cells-per-category ($M/K$) at
         FPR~$= 0.20$. Recall is the mean over 10 paired seeds, with
         SD reported alongside.}
\label{tab:scaling}
\begin{tabular}{rrrrrr}
\toprule
 & \multicolumn{2}{c}{Net-A (537)} & \phantom{a} &
   \multicolumn{2}{c}{Net-C (12{,}861)} \\
\cmidrule{2-3} \cmidrule{5-6}
$K$ & $M/K$ & Recall (mean $\pm$ SD) && $M/K$ & Recall (mean $\pm$ SD) \\
\midrule
3  & 179  & $0.433 \pm 0.004$ && 4{,}287 & $0.696 \pm 0.001$ \\
8  &  67  & $0.506 \pm 0.022$ && 1{,}608 & $0.846 \pm 0.000$ \\
15 &  36  & $0.527 \pm 0.008$ && 857     & $0.807 \pm 0.001$ \\
30 &  18  & $0.538 \pm 0.015$ && 429     & $0.822 \pm 0.006$ \\
\bottomrule
\end{tabular}
\end{table}

Category sizes are highly imbalanced at higher $K$. At $K = 30$, the
largest category in Net-C contains 1{,}498 cells while the smallest
contains 1. Singleton categories offer no statistical basis for
category-level detection. The occupancy ceiling makes the minimum
category size $q$ an explicit design input. Enforcing it from $K = 30$,
by merging the smallest category into the one with the nearest centroid
until every category holds at least $q$ cells, keeps
$H(C) \leq \log_2 K_{\mathrm{eff}} \leq \log_2 \lfloor M/q \rfloor$ in
every case for $q \in \{10, 30\}$. At $q = 30$ merging lowers $H(C)$ by
at most 0.04 bits on Net-B, Net-C and Net-D, because small categories
carry little probability mass, and by 1.08 bits on the small Net-A.
Recall falls by 0.04 to 0.05 on the three large networks and from 0.54
to 0.36 on Net-A. Here $q$ counts contributing cells. Within-category
dependence must be considered when calibrating this requirement to a
detection-power target (\cref{sec:empirical:strict}).

\subsection{Category Construction Cost}
\label{sec:empirical:cost}

\Cref{tab:cost} reports the wall-clock time for category construction
($k$-means fit on per-cell mean profiles, mean $\pm$ SD over 5 runs) and
the single-measurement database query time for aggregating hourly PM
vectors from the raw data.

\begin{table}[t]
\centering
\caption{Category construction cost (seconds). Query is one timed DuckDB
         aggregation of raw records into hourly per-cell PM vectors. Fit is
         mean $\pm$ SD over 5 $k$-means runs on per-cell profiles
         ($n_{\mathrm{init}} = 10$, 11 features).}
\label{tab:cost}
\small
\begin{tabular}{lrrrr}
\toprule
 & Net-A & Net-B & Net-C & Net-D \\
 & (537) & (4{,}987) & (12{,}861) & (6{,}136) \\
\midrule
Query (single)       & 2.7 & 39.9 & 60.4 & 32.3 \\
Fit $K\!=\!8$  & $0.068 \pm 0.003$ & $0.209 \pm 0.028$ & $1.397 \pm 0.050$ & $0.211 \pm 0.016$ \\
Fit $K\!=\!15$ & $0.090 \pm 0.001$ & $0.320 \pm 0.025$ & $2.199 \pm 1.122$ & $0.401 \pm 0.051$ \\
Fit $K\!=\!30$ & $0.125 \pm 0.008$ & $0.419 \pm 0.022$ & $2.365 \pm 1.424$ & $1.032 \pm 0.611$ \\
\bottomrule
\end{tabular}
\end{table}

The $k$-means fit is sub-second for networks up to $\sim$5{,}000 cells
and under 2.5 seconds even for the largest network (12{,}861 cells) at
$K = 30$. The dominant cost is the database query that reduces
millions of raw PM records into hourly cell-level vectors. In a
production deployment, the reduction would be performed incrementally
by the O1 PM pipeline, which eliminates this cost entirely. The fit
itself is negligible, which confirms that category recomputation can be
performed
frequently (hourly or daily) without operational burden.

\subsection{Temporal Robustness}
\label{sec:empirical:staleness}

To assess how quickly categories become stale, we train the category
structure on week~1 (2025-W45) and evaluate detection performance on
each subsequent week without retraining. \Cref{tab:staleness} reports
recall and mean centroid drift (Euclidean distance between the
week-1 centroids and the centroids that would be computed from each
evaluation week's data) at $K = 15$.

\begin{table}[t]
\centering
\caption{Temporal staleness at $K = 15$. Recall is the mean over 10
         paired seeds with SD. Centroid drift is the single-seed
         reference value. Schedule retention uses the single-seed
         reference run. Weekly and biweekly values are projections
         assuming that retraining restores baseline recall. The
         no-recompute row is the observed five-week mean.}
\label{tab:staleness}
\begingroup
\setlength{\tabcolsep}{5pt}
\begin{tabular}{lcccccccc}
\toprule
 & \multicolumn{2}{c}{+1 week} & \multicolumn{2}{c}{+2 weeks}
 & \multicolumn{2}{c}{+3 weeks} & \multicolumn{2}{c}{+4 weeks} \\
\cmidrule(lr){2-3} \cmidrule(lr){4-5} \cmidrule(lr){6-7} \cmidrule(lr){8-9}
Network & Rec.\ ($\pm$ SD) & D. & Rec.\ ($\pm$ SD) & D. & Rec.\ ($\pm$ SD) & D. & Rec.\ ($\pm$ SD) & D. \\
\midrule
Net-A & $0.487 \pm 0.011$ & 1.1 & $0.481 \pm 0.012$ & 1.3 & $0.470 \pm 0.009$ & 1.6 & $0.484 \pm 0.012$ & 1.4 \\
Net-B & $0.738 \pm 0.006$ & 4.2 & $0.726 \pm 0.006$ & 4.8 & $0.642 \pm 0.005$ & 5.1 & $0.527 \pm 0.009$ & 5.6 \\
Net-C & $0.819 \pm 0.001$ & 0.4 & $0.795 \pm 0.001$ & 4.4 & $0.779 \pm 0.001$ & 7.8 & $0.758 \pm 0.002$ & 7.6 \\
Net-D & $0.607 \pm 0.009$ & 0.9 & $0.583 \pm 0.007$ & 1.2 & $0.577 \pm 0.005$ & 1.1 & $0.565 \pm 0.004$ & 1.4 \\
\bottomrule
\end{tabular}
\vspace{0.5em}
\begin{tabular}{lcccc}
\toprule
Schedule retention (mean recall) & Net-A & Net-B & Net-C & Net-D \\
\midrule
Weekly (projected)                  & 0.506 & 0.763 & 0.831 & 0.631 \\
Biweekly (projected)                & 0.490 & 0.751 & 0.826 & 0.618 \\
No recompute (5-week average)       & 0.473 & 0.675 & 0.796 & 0.587 \\
\bottomrule
\end{tabular}
\endgroup
\end{table}

The schedule projections assume that each retraining restores the
initial recall $r_0$. Weekly retention is therefore $r_0$, and biweekly
retention over five weeks is $(3r_0 + 2r_1)/5$, where $r_1$ is recall
one week after the initial training. Three of four networks (Net-A,
Net-C, Net-D) show gradual degradation. Under this assumption, the
biweekly schedule projects 0.490 / 0.826 / 0.618 against the weekly
reference of 0.506 / 0.831 / 0.631. Net-B exhibits a
sharper decline. Its single-week drops in \cref{tab:staleness} are
0.012 (+1 $\to$ +2), 0.084 (+2 $\to$ +3) and 0.115 (+3 $\to$ +4),
so the loss accelerates from +2 onward, and the acceleration survives
seed averaging. On Net-B the projected biweekly retention is 0.012
below the weekly projection. The observed no-recompute mean is 0.088
below that weekly projection (0.675 vs 0.763).

The centroid drift signal tracks the recall loss and provides a
practical monitoring trigger. For each category the monitor computes the
distance between the trained centroid and the mean of the current
window's per-cell mean profiles, from the same standardised O1 counters
that the verifier uses. Its input is therefore cell-level PM data only,
and it satisfies Definition~\ref{def:node_opacity}. The monitor signals that the
category structure has drifted, not why. The operator's response does
not need the cause. It recomputes the categories from the same counters,
which takes seconds (\cref{sec:empirical:cost}). Diagnosing the cause of
a traffic shift, where an operator needs it, is a separate
fault-management task outside the assurance loop. On Net-B the mean
centroid drift
at +1 week is 4.24 against at most 1.12 elsewhere at the same offset,
so a drift threshold near 3 units cleanly separates Net-B's
nonstationarity from the other networks and would trigger a
recomputation before the +3 cliff. Given that category construction
takes under 3 seconds even for the largest network
(\cref{sec:empirical:cost}), the operator response is trivially
cheap.

The corresponding \texttt{administrativeState} lock-event stream is
uninformative for this drift. On Net-B the standardised CM lock rate
is at most 0.0009 of cell-hours per week during the study window.
The lock timeline shows no coincident
transition with the +3 recall cliff. The traffic-profile shift that
degrades detection is invisible in the CM channel. A PM-side drift
monitor is therefore a necessary (not merely convenient) part of the
deployed loop. Subject to the assumed restoration of baseline recall,
the projections support a biweekly schedule on Net-A, Net-C and Net-D
and a weekly schedule on Net-B. Either schedule should be augmented
with the PM-side drift trigger to catch mid-cycle shifts.

\subsection[Sensitivity to alpha-zero and the Exceedance Alarm]{Sensitivity to $\alpha_0$ and the Exceedance Alarm}
\label{sec:empirical:alpha0}

The per-cell false positive target $\alpha_0$ is an internal
configuration knob. The operator-facing rate is the category-level
alarm rate, calibrated separately as a single network-level quantile
of the pooled validation per-category flagged-fraction distribution
(\cref{sec:deployment:impl}). The per-cell threshold is applied
uniformly across categories, so per-category empirical null rates
vary with category tightness even though the network-wide rate meets
$\alpha_0$ by construction. We sweep $\alpha_0 \in \{0.05, 0.10,
0.20\}$ and $K \in \{8, 15\}$ over the four networks.
\Cref{tab:alpha0sweep} reports recall and
empirical per-cell FPR at each grid point.

\begin{table}[t]
\centering
\caption{Per-cell recall and empirical FPR across $\alpha_0 \in \{0.05,
         0.10, 0.20\}$ and $K \in \{8, 15\}$.}
\label{tab:alpha0sweep}
\begin{tabular}{lcc cc cc}
\toprule
 & \multicolumn{2}{c}{$\alpha_0=0.05$} & \multicolumn{2}{c}{$\alpha_0=0.10$} & \multicolumn{2}{c}{$\alpha_0=0.20$} \\
\cmidrule(lr){2-3} \cmidrule(lr){4-5} \cmidrule(lr){6-7}
Network ($K$) & Recall & FPR & Recall & FPR & Recall & FPR \\
\midrule
Net-A ($K{=}8$)  & 0.263 & 0.049 & 0.375 & 0.097 & 0.534 & 0.194 \\
Net-A ($K{=}15$) & 0.226 & 0.048 & 0.331 & 0.095 & 0.493 & 0.193 \\
Net-B ($K{=}8$)  & 0.307 & 0.075 & 0.458 & 0.148 & 0.656 & 0.280 \\
Net-B ($K{=}15$) & 0.367 & 0.078 & 0.509 & 0.150 & 0.689 & 0.284 \\
Net-C ($K{=}8$)  & 0.567 & 0.053 & 0.687 & 0.106 & 0.845 & 0.212 \\
Net-C ($K{=}15$) & 0.520 & 0.054 & 0.654 & 0.109 & 0.807 & 0.216 \\
Net-D ($K{=}8$)  & 0.359 & 0.049 & 0.503 & 0.099 & 0.648 & 0.200 \\
Net-D ($K{=}15$) & 0.338 & 0.049 & 0.473 & 0.099 & 0.627 & 0.201 \\
\bottomrule
\end{tabular}
\end{table}

Empirical per-cell FPR matches the calibration target $\alpha_0$
within 0.02 for Net-A, Net-C, and Net-D across the whole grid, with
the largest deviations at $\alpha_0 = 0.20$ (Net-A 0.007 at $K = 15$,
Net-C 0.016 at $K = 15$, Net-D 0.001 at $K = 15$). On
Net-B the empirical FPR overshoots the target at every grid point
(0.284 against 0.20 at $\alpha_0=0.20$, $K=15$, and 0.078 against
0.05 at $\alpha_0=0.05$, $K=15$), consistent with the temporal
nonstationarity finding in \cref{sec:empirical:staleness}. In
deployment the operator would recalibrate on a rolling validation
window to remove this bias. The paper reports raw empirical FPR
without any post-hoc correction.

The category-level alarm operates on top of the per-cell decisions.
For each network we characterise its null distribution on validation
(no injection) and its response under injected load.
At $K = 15$ the category-level mean flagged-fraction shifts from
null to test by a factor of 1.24--1.48 across the four networks and
$\alpha_0$ values, with Pearson $r > 0.98$ against $\alpha_0$ on the
null. The distribution shift is monotone in $\alpha_0$ within
network. A tail-quantile alarm calibrated at category-level null
rates 0.05 or 0.01 saturates on three of four networks (Net-B, Net-C,
Net-D). The null validation distribution already assigns mass at
flagged-fraction $= 1$ to a substantial fraction of small
category-hours, so the tail alarm cannot fire more often on test.
The saturation mechanism is the uniform per-cell threshold. Loose or
heavy-tailed categories realise per-category null flag rates well
above $\alpha_0$ even though the pooled network-wide rate meets
$\alpha_0$, which drives the pooled per-category flagged-fraction
distribution into the ceiling.
The mean flagged-fraction separation (1.24 to 1.48 times at $K = 15$)
is thus the surviving category-level discriminant on the injected
load. The tail-exceedance parameterisation saturates on three of four
networks at the tested operating point. Operationalising this
discriminant as a sequential decision rule is deployment work outside
the scope of this paper. Both the internal $\alpha_0$ and the
category-level alarm rate are deployer-tunable independently.

\subsection{Per-Cell and Category-Mean Verification}
\label{sec:empirical:strict}

The strict category-mean verifier computes one distance per
(category, period), from the category mean
$\bar{\bm{X}}_k = (1/m_k) \sum_{i\in k} \bm{X}_i$ to the trained
centroid. It reduces the interface exposure from all cell-level counters
to $K$ scalar statistics per period. To compare it with per-cell
screening, both verifiers are scored on one decision unit, the
(category, hour) pair, under one injection. At $K = 15$, 30\% of test
pairs are degraded. Inside a degraded pair a fraction $w$ of the cells
is replaced by training cell-hour vectors from other categories. Two
regimes vary the extent and source of replacement. In the
\emph{category-wide replacement} regime every cell of the pair moves
($w = 1$), and all replacements come from one other category. This
regime is labelled ``common'' in the tables. In the
\emph{sparse} regime 30\% of the cells move ($w = 0.3$), each toward an
independently drawn category, which is the cell-level design of
\cref{sec:empirical:instantiation} applied inside a pair. The per-cell
verifier is the exceedance statistic of \cref{sec:empirical:alpha0},
the fraction of a pair's cells beyond the per-cell threshold. Each
statistic is standardised per category on the validation period, and
\cref{tab:strict} reports the threshold-free area under the ROC curve
(AUC) over 10 paired seeds.

\begin{table}[t]
\centering
\caption{Category-mean vs per-cell exceedance verification on one
         decision unit at $K = 15$. AUC after per-category
         standardisation, mean $\pm$ SD over 10 paired seeds.}
\label{tab:strict}
\small
\begin{tabular}{lcccc}
\toprule
 & \multicolumn{2}{c}{Category-wide ($w = 1$)}
 & \multicolumn{2}{c}{Sparse ($w = 0.3$)} \\
\cmidrule(lr){2-3} \cmidrule(lr){4-5}
Network & Category mean & Per-cell & Category mean & Per-cell \\
\midrule
Net-A & $0.848 \pm 0.008$ & $0.851 \pm 0.007$ & $0.504 \pm 0.012$ & $0.710 \pm 0.011$ \\
Net-B & $0.869 \pm 0.009$ & $0.752 \pm 0.013$ & $0.750 \pm 0.010$ & $0.715 \pm 0.013$ \\
Net-C & $0.939 \pm 0.003$ & $0.855 \pm 0.008$ & $0.782 \pm 0.008$ & $0.776 \pm 0.009$ \\
Net-D & $0.954 \pm 0.004$ & $0.914 \pm 0.004$ & $0.651 \pm 0.013$ & $0.811 \pm 0.008$ \\
\bottomrule
\end{tabular}
\end{table}

Under category-wide replacement, the category-mean score has higher AUC
on Net-B, Net-C and Net-D and comparable AUC on Net-A. These experiments
compare two implemented scores. They do not establish the
state-independent observation premise or the optimal-test power
ordering of Proposition~\ref{prop:node_opacity}. Under sparse degradation the
ordering varies by network. Per-cell exceedance is better on Net-A
(0.710 against 0.504) and Net-D (0.811 against 0.651). The category mean
is better on Net-B (0.750 against 0.715), with a smaller difference on
Net-C (0.782 against 0.776).

\emph{Calibration.} Category-level statistics need per-category
thresholds. The hourly category mean of a large, homogeneous category
varies little, while that of a small or heterogeneous category varies a
lot. One quantile pooled across categories is set by the latter, so the
former rarely fire. Under whole-aggregate replacement at $K = 15$ with a
pooled threshold, category-mean recall on Net-B and Net-C equals the
false-alarm rate (0.228 against 0.214 and 0.201 against 0.206), so the
verifier is at chance. After per-category standardisation the same
statistic reaches AUC 0.889 and 0.971 (0.851 and 0.967 on Net-A and
Net-D). With per-category thresholds calibrated on the validation
period, the realised false-alarm rate on Net-B and Net-C exceeds the
0.20 target (0.37 and 0.26 under category-wide replacement). This is the
nonstationarity of \cref{sec:empirical:staleness}, and the rolling
recalibration of \cref{sec:empirical:alpha0} applies.

\begin{table}[t]
\centering
\caption{Mechanism diagnostics at $K = 15$, mean over 10 paired seeds.
         $m$ is the category size. VIF is the variance of the hourly
         category mean relative to $\sigma^2_{\mathrm{cell}}/m$ on the
         validation period, and $m_{\mathrm{eff}} = m/\mathrm{VIF}$ the
         effective size. Each of $m$, VIF and $m_{\mathrm{eff}}$ is the
         median over categories, so $m_{\mathrm{eff}}$ is not the ratio
         of the two displayed medians. ``Shuffled'' permutes each cell's hours independently,
         which removes co-movement within an hour. Crowding is the
         nearest-centroid distance over the median null distance.
         Cancellation is $\|\sum_i \bm{\Delta}_i\| / \sum_i
         \|\bm{\Delta}_i\|$ over the replaced cells of a pair (common /
         sparse).}
\label{tab:mechanisms}
\small
\begin{tabular}{lrrrrcrc}
\toprule
 & & & & VIF & Sparse AUC & & \\
Network & $m$ & VIF & $m_{\mathrm{eff}}$ & shuffled & (shuffled) &
  Crowding & Cancellation \\
\midrule
Net-A & 29 & 4.2 & 6.2 & 0.93 & 0.504 $\to$ 0.754 & 0.56 & 0.66 / 0.58 \\
Net-B & 242 & 20.9 & 6.4 & 0.89 & 0.750 $\to$ 0.729 & 0.85 & 0.83 / 0.66 \\
Net-C & 871 & 26.1 & 5.9 & 0.78 & 0.782 $\to$ 0.916 & 1.06 & 0.89 / 0.56 \\
Net-D & 293 & 24.3 & 14.3 & 0.87 & 0.651 $\to$ 0.935 & 0.79 & 0.72 / 0.45 \\
\bottomrule
\end{tabular}
\end{table}

\Cref{tab:mechanisms} measures the three mechanisms that
Remark~\ref{rem:tworegimes} names. Co-movement within an hour dominates. The
variance of the category mean is 4 to 26 times what independent cells
would give, so categories of 29 to 871 cells behave like 6 to 14
independent cells. With within-period correlation the effective size of
a category mean stays near the reciprocal of that correlation, however
many cells the category holds. Shuffling each cell's hours removes the
inflation and raises the sparse-regime AUC of the category mean on
Net-A, Net-C and Net-D. On Net-B it does not, and the week-to-week
nonstationarity of \cref{sec:empirical:staleness} is the more likely
limit. Crowding is strongest on Net-A, where the nearest other centroid
lies at 0.56 of the median null distance, and Net-A has the lowest
sparse-regime AUC. Cancellation is secondary. Replacements toward
independent categories keep 45 to 66\% of their summed displacement,
against 66 to 89\% for a common source.

The practical guidance follows. Use the category mean where degradation
is expected to affect a category as a whole, and per-cell exceedance for
faults that hit some cells. Both read the same O1 counters, so an
operator can run both. Calibrate each per category. Judge category-mean
power by the effective size $m_{\mathrm{eff}}$, which can be estimated on
validation data from O1 counters alone, rather than by the nominal
category size.

We further characterise the utility cost of differential privacy at
the interface actually released to a downstream assurance consumer.
For category $k$ and period $t$, the trusted operator-side verifier
first computes the deterministic MISES alarm
\begin{equation}
 b_{k,t}=\mathbf{1}\{d(\bar{\bm X}_{k,t},\bm\mu_k)>\tau_k\}.
\end{equation}
The fitted category map, scaler, centroids, and threshold are fixed
before the evaluated assurance window and remain internal. Symmetric
randomized response releases $z_{k,t}=b_{k,t}$ with probability
$p_\varepsilon=e^\varepsilon/(1+e^\varepsilon)$ and releases
$1-b_{k,t}$ otherwise.

The trusted-curator contract protects one assurance-time cell-period
contribution against a downstream consumer that observes only $z_{k,t}$
or its post-processing. For any neighbouring inputs to a query, the
deterministic alarms are either equal, in which case the output laws
coincide, or opposite, in which case the likelihood ratio for either
released bit is at most $p_\varepsilon/(1-p_\varepsilon)=e^\varepsilon$.
The release is therefore pure $\varepsilon$-DP ($\delta=0$), independent
of the detector's internal sensitivity. With a fixed map, one
cell-period contributes to at most one category query, so disjoint
category releases within a period compose in parallel. Repeated
periods involving the same protected cell compose sequentially. Under
basic composition, $T$ releases cost $T\varepsilon$. Downstream actions
computed solely from the privatized bits are post-processing and do
not remove that repeated-release cost.

\begin{table}[t]
\centering
\caption{Pure-DP randomized response on the category-mean alarm with
         per-category thresholds, $\delta = 0$. Unprivatized recall,
         false-alarm rate (FPR), and their difference $J$, then $J$ at
         $\varepsilon \in \{1, 2, 5\}$. Rates are means over 10 paired
         seeds. Privatized values are exact.}
\label{tab:dpsummary}
\small
\begin{tabular}{llccccccc}
\toprule
 & & \multicolumn{3}{c}{Unprivatized} & \multicolumn{3}{c}{$J$ under
   randomized response} \\
\cmidrule(lr){3-5} \cmidrule(lr){6-8}
Network & Regime & Recall & FPR & $J$ & $\varepsilon=1$ & $\varepsilon=2$
  & $\varepsilon=5$ \\
\midrule
Net-A & common & 0.699 & 0.187 & 0.512 & 0.237 & 0.390 & 0.505 \\
Net-A & sparse & 0.170 & 0.185 & -0.016 & -0.007 & -0.012 & -0.015 \\
Net-B & common & 0.899 & 0.371 & 0.528 & 0.244 & 0.402 & 0.521 \\
Net-B & sparse & 0.747 & 0.370 & 0.377 & 0.174 & 0.287 & 0.372 \\
Net-C & common & 0.902 & 0.259 & 0.643 & 0.297 & 0.490 & 0.635 \\
Net-C & sparse & 0.665 & 0.263 & 0.403 & 0.186 & 0.307 & 0.397 \\
Net-D & common & 0.924 & 0.192 & 0.732 & 0.338 & 0.558 & 0.722 \\
Net-D & sparse & 0.417 & 0.193 & 0.225 & 0.104 & 0.171 & 0.222 \\
\bottomrule
\end{tabular}
\end{table}

If unprivatized recall is $r$, unprivatized FPR is $a$, and
$q_\varepsilon=1/(1+e^\varepsilon)$, the exact privatized rates are
$q_\varepsilon+(1-2q_\varepsilon)r$ and
$q_\varepsilon+(1-2q_\varepsilon)a$. Randomized response therefore
moves both rates toward $0.5$ as privacy strengthens. Recall alone can
mislead. Where unprivatized recall is below $0.5$, stronger
randomization raises recall, but it raises FPR by the same mechanism.
The quantity to read is $J = r - a$, the detector's margin over chance.
Randomized response multiplies it by
$1 - 2q_\varepsilon = \tanh(\varepsilon/2)$, which is 0.46 at
$\varepsilon = 1$, 0.76 at $\varepsilon = 2$ and 0.99 at
$\varepsilon = 5$. It can only shrink the size of the margin, and a negative margin moves toward zero, not above it. Where the
unprivatized margin is already near zero, as for the category mean
under sparse degradation on Net-A, randomization cannot help, and the
operator should change verifier rather than privatize one that does
not discriminate. Differential privacy protects a contribution against
inference from the alarm release, whereas the MISES $\log_2 K$ bound
caps information flow through the category channel.
\Cref{tab:dpsummary} compares utility, not equivalent privacy
definitions. The federated-learning comparison in
\cref{sec:empirical:footprint} remains qualitative by design. There
is no training pipeline to federate, and FL protects model training
rather than assurance-time observation.

\subsection{Structured Degradation Scenarios}
\label{sec:empirical:scenarios}

Cross-category replacement tests whether a verifier notices that a
cell-hour no longer looks like its category. Real degradations move
specific counters in specific directions. The four datasets carry no
fault-management labels (\cref{sec:empirical:realvalidation}), so this
section stays synthetic, but it replaces the cross-category draw with
counter-level perturbations shaped like three common events. Each acts
on the raw counters of the affected cells with magnitude
$f \in \{1.5, 3\}$. An \emph{outage}, a proxy for a transport cut or a
site-power fault, divides traffic volume, active UEs, random-access
successes and PRB utilisation by $f$. \emph{Interference} multiplies the
HARQ NACK ratio, the QPSK share and RLC delay by $f$, divides the 256QAM
share by $f$, and divides downlink volume by $\sqrt{f}$. A \emph{surge}
multiplies active UEs, traffic volume and random-access successes by
$f$, PRB utilisation by $f$ up to 1, and RLC delay by $\sqrt{f}$. The
injection, calibration and seeds follow \cref{sec:empirical:strict}.

\begin{table}[t]
\centering
\caption{Structured degradations at $K = 15$, mean over 10 paired seeds.
         Per-cell recall and FPR are cell-level, with 30\% of the cells in
         degraded pairs perturbed. Category-mean AUC is for the
         category-wide regime ($w = 1$, labelled ``common''). The last two
         columns are per-cell exceedance AUC in the same regime, the comparator
         for the category mean.}
\label{tab:scenarios}
\footnotesize
\setlength{\tabcolsep}{3.5pt}
\begin{tabular}{llccccccc}
\toprule
 & & \multicolumn{3}{c}{Per-cell, sparse} &
   \multicolumn{2}{c}{Category mean, common} &
   \multicolumn{2}{c}{Exceedance, common} \\
\cmidrule(lr){3-5} \cmidrule(lr){6-7} \cmidrule(lr){8-9}
Network & Scenario & $f=1.5$ & $f=3$ & FPR & $f=1.5$ & $f=3$ & $f=1.5$ & $f=3$ \\
\midrule
Net-A & outage & 0.181 & 0.244 & 0.193 & 0.500 & 0.754 & 0.447 & 0.538 \\
 & interference & 0.299 & 0.960 & 0.193 & 0.797 & 0.971 & 0.732 & 0.971 \\
 & surge & 0.314 & 0.857 & 0.193 & 0.754 & 0.941 & 0.692 & 0.945 \\
Net-B & outage & 0.256 & 0.385 & 0.284 & 0.519 & 0.776 & 0.433 & 0.605 \\
 & interference & 0.362 & 0.669 & 0.284 & 0.686 & 0.814 & 0.687 & 0.801 \\
 & surge & 0.448 & 0.951 & 0.284 & 0.729 & 0.927 & 0.704 & 0.803 \\
Net-C & outage & 0.198 & 0.242 & 0.216 & 0.499 & 0.775 & 0.428 & 0.550 \\
 & interference & 0.297 & 0.747 & 0.216 & 0.766 & 0.915 & 0.669 & 0.866 \\
 & surge & 0.295 & 0.471 & 0.216 & 0.725 & 0.850 & 0.712 & 0.834 \\
Net-D & outage & 0.184 & 0.226 & 0.201 & 0.540 & 0.755 & 0.419 & 0.519 \\
 & interference & 0.290 & 0.985 & 0.201 & 0.800 & 0.962 & 0.719 & 0.889 \\
 & surge & 0.285 & 0.547 & 0.201 & 0.681 & 0.916 & 0.704 & 0.904 \\
\bottomrule
\end{tabular}
\end{table}

\Cref{tab:scenarios} shows that interference and surge at $f = 3$ are
visible to the deployed per-cell screener. Recall reaches 0.67 to 0.99
for interference and 0.47 to 0.95 for surge, against FPR between 0.19
and 0.28. At $f = 1.5$ both are harder. For a category-wide version of either,
the category mean (AUC 0.68 to 0.80) beats per-cell exceedance (0.67 to
0.73) in six of the eight network and scenario cases and ties in a
seventh at the displayed precision. This is an empirical comparison of
the two scores, as in \cref{sec:empirical:strict}. Outages are the weak case. Per-cell recall stays within
0.06 of the FPR on Net-A, Net-C and Net-D even at $f = 3$, because a cell that
loses traffic looks like the same cell in a quiet hour, and quiet hours
lie inside the null distribution of distance to the centroid. A
category-wide outage at $f = 3$ is still caught by the category mean
(AUC 0.75 to 0.78), but a sparse one is not. Distance to a category
centroid is a two-sided, direction-free test, and traffic loss calls for
a one-sided one.

\subsection{Real-Label Category Validation}
\label{sec:empirical:realvalidation}

The injection design (\cref{sec:empirical:instantiation}) provides
ground truth for the detection task but not for whether the learned
category structure has real operational meaning. As a
feature-disjoint real-label check we compare the learned categories
against the O1 \texttt{administrativeState} lock stream, a signal
that was not used at training time or inside any of the eleven
clustering features.

On Net-C one cluster in the $K = 15$ partition (cluster~4, 1,715
cells) carries $92.3\%$ of the partition's locked cell-hours, and
$65.1\%$ of the cluster's own cell-hours are locked. The fourteen
other clusters carry $\leq 1.72\%$ of their own cell-hours locked
each. The
clustering thus discovers a coherent standing-lock population from
PM signals alone. Over cell-hours in the study window the partition
carries $0.261$ bits of mutual information with the lock label,
$60.3\%$ of the label's $0.433$-bit entropy. Normalised mutual
information under the standard arithmetic-mean normalisation is
$0.14$, which reflects the imbalance between the binary label and the
$3.24$-bit partition. Both readings are consistent with one strongly
aligned cluster among fifteen and no dependence on the labels for
the remaining fourteen.
This is a partial real-label validation of the category structure
without any promise of real-event detection recall (the standing
population is not an event stream).

\subsection{Data Footprint Comparison}
\label{sec:empirical:footprint}

\Cref{tab:footprint} compares the data requirements of the MISES
architecture against two alternative intent assurance approaches, the
GenAI-IDN architecture of Habib et al.\ \cite{habib2025},
which processes per-flow session records, and an E2-based approach
using per-UE near-real-time RIC telemetry.

\begin{table}[t]
\centering
\caption{Data footprint per cell per reporting period. Volume ratio
         is relative to MISES. The MISES column is the per-cell
         verifier used in the paper's headline.
         \Cref{sec:empirical:strict} additionally measures the strict
         category-mean variant, which reduces the interface exposure
         further to $K$ scalar statistics per period.}
\label{tab:footprint}
\small
\begin{tabular}{lccc}
\toprule
 & MISES & GenAI-IDN & E2-based \\
\midrule
Granularity    & Cell agg. & Per-flow  & Per-UE \\
Interface      & O1        & Custom    & E2 \\
Records        & 1         & ${\sim}10^3$ & ${\sim}10^2$ \\
Features       & 11        & 12        & 15 \\
Scalars        & 11        & 12{,}000  & 1{,}500 \\
Vol.\ ratio    & $1\times$ & ${\sim}10^3\!\times$ & ${\sim}140\!\times$ \\
Per-UE data    & No        & Yes       & Yes \\
DPIA scope     & Reduced   & Full      & Full \\
Vendor portable & Yes      & No        & Partial \\
\bottomrule
\end{tabular}
\end{table}

The MISES architecture requires three orders of magnitude less data per
cell per reporting period than the GenAI-IDN approach and two orders
less than E2-based assurance. The reduction is not a compression or
approximation. The architecture never collects the per-flow or per-UE
data in the first place. This elimination, rather than reduction, of
data collection is what produces the privacy guarantee
(\cref{thm:dpi}) and narrows the scope of any data protection
impact assessment under GDPR. Whether a DPIA is required at all is a
controller-specific determination under Article~35. The per-flow and
per-UE collection that commonly triggers one is absent here.

The O1 interface dependency means the architecture is deployable in any
O-RAN-compliant network regardless of vendor-specific E2 exposure. The
GenAI-IDN approach requires custom flow-level export capabilities that
no current O-RAN specification mandates. The E2 approach depends on
per-UE indication messages whose availability and schema vary across
RAN vendors.

\subsection{Implications for Intent Category Design}
\label{sec:empirical:design}

The results inform the choice of $K$ in three steps. First, fix the
minimum number of contributing cells $q$ that the chosen verifier
needs. For the category mean, calibrate this requirement to the power
target using the observed within-category dependence
(\cref{tab:mechanisms}). This sets the occupancy ceiling $\lfloor M/q \rfloor$ and, by
Corollary~\ref{cor:capped}, the leakage ceiling. Second, along a nested family
of partitions, the welfare and least-category detection curves give the
band where both targets are met (Theorem~3 of \cite{mises_preprint}),
provided both orderings hold on that family. Third, choose $K$ inside the band
according to the operator's priority between provisioning precision,
assurance confidence and privacy.

On the four networks tested, independent $k$-means partitions with
$K = 5$ to 8 sit at the start of the detection plateau
(\cref{sec:empirical:plateau}). Beyond $K = 8$ additional categories
improve provisioning while detection gains nothing and the leakage
ceiling grows. This is an observation on these networks, not a
certified range. In the 3GPP intent lifecycle the intent translation
stage should map intents to a moderate number of categories and check
the choice against measured curves, as above.

% =======================================================================
\section{Discussion}
\label{sec:discussion}
% =======================================================================

\subsection{Architectural Privacy vs.\ Mechanism Privacy}

The MISES architecture achieves privacy by not collecting data, rather
than by protecting collected data. This is a structural distinction from
differential privacy, federated learning, and secure computation, all of
which assume a full-information pipeline exists and then limit what can
be inferred from it. The leakage bound (\cref{thm:dpi}) formalises this
distinction. The ceiling on what the resource allocation reveals is set
by the category architecture itself, not by a post-hoc protection
mechanism.

This distinction has practical consequences. Mechanism-based privacy
(DP, FL, HE) requires ongoing operational effort in noise calibration,
gradient clipping, and key management. Architectural privacy needs none
of these once the information flow is fixed. It does need assurance that
the flow is not bypassed, which \cref{sec:limitations} addresses.

\subsection{Threat Model and Limitations}
\label{sec:limitations}

\subsubsection{Formal threat model.}
The adversary is an honest-but-curious coordinator that follows the
protocol and observes exactly the two channels the architecture
exposes. These are the category signal $C \in \{1, \ldots, K\}$
(equivalently, the resource allocation $R = \phi(C)$) at the
provisioning interface,
and the cell-level standardised O1 PM counters at the verification
interface. The coordinator does not deviate from the protocol, does
not deploy DPI, does not read per-UE or per-flow traces, and does not
access vendor-internal node state. The category-channel bound limits
what the declaration and resource allocation reveal. It does not bound
the coordinator's total inference from the PM counters. Writing $X$ for
those observations, the additional category disclosure satisfies
$I(T;C\mid X) \leq H(C\mid X) \leq H(C) \leq \log_2 K$,
while $I(T;C,X)$ also includes $I(T;X)$, which this result does not bound.
Node-opacity restricts the telemetry interface without assigning a numerical ceiling
to inference from that telemetry. The following classes of adversary
are explicitly out of scope:
\begin{enumerate}[leftmargin=*,nosep,label=(\alph*)]
\item \emph{Side-channel adversaries} that correlate cell-level PM
  patterns with lower-layer observations (timing, packet sizes,
  fingerprints) to re-identify traffic content
  \cite{sirinam2018,wang2014}. The leakage bound in \cref{thm:dpi}
  bounds only the information carried through the category channel.
  Side channels sit outside the mechanism's scope.
\item \emph{Auxiliary-data adversaries} with access to external data
  sources correlated with $T$ (external content classifiers,
  device-side signals). Correlated auxiliary data can weaken
  effective privacy even though the mechanism's information flow is
  unchanged.
\item \emph{Active or malicious coordinators} that deviate from the
  protocol (e.g.\ instrument the O1 pipeline to expose per-UE data,
  or deploy DPI outside the mechanism). Architectural privacy
  guarantees hold as long as the architecture is not bypassed. The
  present analysis is not a game-based proof of integrity.
\end{enumerate}

\subsubsection{Residual risk in cloud-hosted and multi-tenant deployments.}
The honest-but-curious coordinator is the operator's own management
function, and the operator is the data controller. Hosting the Non-RT
RIC on third-party cloud infrastructure adds the infrastructure
provider, which the model above does not cover. The architecture cannot
prove from inside the mechanism that it has not been bypassed, but its
narrow interface makes the claim auditable. The coordinator's inputs are
the declared categories, the A1 policy it emits and the O1 PM
measurement types it subscribes to, so an auditor can check that only
cell-level PM types are requested. Remote attestation of the rApp, with
confidential-computing isolation on an untrusted host, binds the
deployed code to the audited version. DPI also needs user-plane access,
which the Non-RT RIC does not have. These measures are not a
cryptographic proof against a malicious coordinator. They make the
assumption checkable. \Cref{tab:residual} summarises how the excluded
threats interact with the dual-purpose category signal when several
tenants share cells.

\begin{table}[t]
\centering
\caption{Residual risk surface. ``Holds'' means the stated bound is
         unchanged. ``Weakens'' means the threat can extract information
         that the bound does not cover.}
\label{tab:residual}
\small
\begin{tabular}{p{0.24\linewidth}p{0.31\linewidth}p{0.35\linewidth}}
\toprule
Threat & Effect on the bounds & Mitigation \\
\midrule
Other tenants on shared cells
  & Holds for each tenant's category channel. Joint category-channel
    leakage over tenants $A$ is at most
    $\sum_A \log_2 \lfloor M_A/q \rfloor$, with $M_A$ the cells that
    carry tenant $A$.
    Shared cell counters mix tenants' load, so one tenant's
    alarms depend on co-located traffic.
  & Release per-tenant alarms only. Randomised response on the alarm
    bit (\cref{sec:empirical:strict}) bounds what a tenant learns from
    it. \\
Auxiliary data correlated with $T$
  & Weakens. The bound covers the category channel, not what $C$
    combined with outside data reveals.
  & Demand-derived rather than semantic categories reduce
    $I(T;\, C)$ below $H(C)$. \\
Side channels below the cell level
  & Weakens. Outside the mechanism.
  & Unchanged by the architecture. The coordinator holds no lower-layer
    data to correlate. \\
Repeated re-declaration
  & Composes. Cumulative leakage is at most $\sum_t H(C_t)$.
  & Retain the intent's category value. Repetition and fixed, public
    post-processing add no category-channel information. \\
Malicious or bypassed coordinator
  & Fails.
  & Interface audit, attestation, and confidential computing, as
    above. \\
\bottomrule
\end{tabular}
\end{table}

\emph{Excluded evaluation scope.} The empirical section uses
injection-based ground truth (\cref{sec:empirical:instantiation}). The
four operator datasets contain no fault-management or alarm labels,
and the O1 \texttt{administrativeState} lock stream is structurally
unsuitable for the hourly detection task. On Net-B its transitions
have daily resolution and standing lock exposure is below 0.001 of
cell-hours in the study window. On Net-C cluster~4 is a
standing-locked population rather than an event stream
(\cref{sec:empirical:realvalidation}). The structured scenarios of
\cref{sec:empirical:scenarios} move counters in the direction real
events would, but they remain synthetic. Any real-event recall claim
therefore lies outside the paper's evidentiary base. This is an
evaluation-scope limitation, not an adversary class.

\subsubsection[Empirical category entropy tightens the leakage bound]{Empirical $H(C)$ tightens the leakage bound.}
The $\log_2 K$ ceiling in \cref{thm:dpi} assumes a uniform category
prior. The learned partitions on the four operator networks are
non-uniform. At $K = 15$, empirical category entropies are 3.56,
3.16, 3.24, and 3.22 bits for Net-A, Net-B, Net-C, and Net-D
respectively, against the uniform ceiling $\log_2 15 = 3.91$ bits,
giving efficiencies of 0.912, 0.810, 0.828, and 0.825 respectively.
The operative ceiling on these deployments is
therefore materially below the uniform $\log_2 K$ value, a
tightening of 8--20\% depending on the network. The measured
quantity is the empirical category entropy $H(C)$. It upper-bounds
the leakage $I(T;R)$ through the category channel and does not
estimate the leakage itself, which may be far lower. Where a deployment
targets a specific bound the operator can enforce it by rebalancing
the partition, at the cost of within-category variance.

\subsubsection{Detection-side qualifications.}
Proposition~\ref{prop:node_opacity} requires state-independent observation
(Remark~\ref{rem:tworegimes}). Category-wide replacement does not by itself
establish that premise. The implemented category-mean score has higher
AUC on three networks and comparable AUC on the fourth in that regime.
Under sparse degradation the ordering varies by network
(\cref{sec:empirical:strict}). Within-period co-movement keeps the
effective size of a category mean near 6 to 14 cells, so the
category mean gains little from large categories. Distance to a
category centroid misses outages that affect only some cells
(\cref{sec:empirical:scenarios}). The empirical illustration uses
injection-based ground truth, not naturally occurring degradation, and
it is not a proof of the theorems.

% =======================================================================
\section{Conclusion}
\label{sec:conclusion}
% =======================================================================

We presented an architecture for privacy-preserving intent fulfilment
and assurance in 6G RAN that operates entirely on declared intent
categories and standardised cell-level O1 PM counters. The category
declaration and resource allocation reveal at most $\log_2 K$ bits
about traffic content, a ceiling set by
the information flow rather than by a post-hoc protection mechanism.
Intent-traffic unlinkability and node-opaque verification hold by
construction. Verification on the category aggregate loses nothing to per-agent
verification when individual records add state-independent noise to
that aggregate.

The central structural finding is that detection and privacy are aligned.
Both favour fewer, larger intent categories, while provisioning quality
opposes both. The minimum category size that verification needs caps the
category entropy at $\log_2 \lfloor M/q \rfloor$, so the detection
requirement fixes the leakage ceiling without a separate privacy
parameter for the category channel. Production data from four operator
networks illustrate this tradeoff and give practical guidance. Beyond
moderate category granularity, provisioning continues to improve while
detection plateaus. The category mean suits category-wide degradation
and per-cell screening suits sparse faults, and both need per-category
calibration.

The deployment mapping to the 3GPP intent lifecycle and O-RAN Non-RT RIC
identifies the concrete interfaces, data objects, and deployment points
at which the architecture operates. The mechanism requires only O1
cell-level PM data, which makes it deployable in any O-RAN configuration
regardless of vendor-specific E2 exposure. The architecture supports
GDPR data minimisation and avoids content interception under the
ePrivacy Directive by construction, as the
management system never collects the data that would require protection.
Full regulatory compliance remains a controller-specific assessment.

% =======================================================================
\section*{Acknowledgements}
% =======================================================================

The author drafted this manuscript with assistance from AI writing and
coding tools. All analysis, results, and text remain the author's
responsibility.

% -----------------------------------------------------------------------
\bibliography{refs}

@misc{mises_preprint,
  author       = {Armstrong, Joss},
  title        = {Information Requirements for Service Allocation and
                  Aggregate Verification},
  year         = {2026},
  doi          = {10.48550/arXiv.2604.26808},
  note         = {arXiv preprint arXiv:2604.26808,
                  \url{https://arxiv.org/abs/2604.26808}},
}

@inproceedings{dwork2006,
  author    = {Dwork, Cynthia},
  title     = {Differential Privacy},
  booktitle = {Proc. 33rd International Colloquium on Automata, Languages and Programming (ICALP)},
  pages     = {1--12},
  year      = {2006},
  publisher = {Springer},
  address   = {Venice, Italy},
  doi       = {10.1007/11787006_1},
}

@inproceedings{mcmahan2017,
  author    = {McMahan, Brendan and Moore, Eider and Ramage, Daniel
               and Hampson, Seth and Arcas, Blaise Ag\"{u}era y},
  title     = {Communication-Efficient Learning of Deep Networks
               from Decentralized Data},
  booktitle = {Proc. 20th International Conference on Artificial
               Intelligence and Statistics (AISTATS)},
  pages     = {1273--1282},
  year      = {2017},
}

@phdthesis{gentry2009,
  author = {Gentry, Craig},
  title  = {A Fully Homomorphic Encryption Scheme},
  school = {Stanford University},
  year   = {2009},
}

@misc{cavoukian2009,
  author = {Cavoukian, Ann},
  title  = {Privacy by Design: The 7 Foundational Principles},
  year   = {2009},
  note   = {Information and Privacy Commissioner of Ontario},
}

@inproceedings{abadi2016,
  author    = {Abadi, Martin and Chu, Andy and Goodfellow, Ian and
               McMahan, H. Brendan and Mironov, Ilya and Talwar, Kunal
               and Zhang, Li},
  title     = {Deep Learning with Differential Privacy},
  booktitle = {Proc. ACM SIGSAC Conference on Computer and
               Communications Security (CCS)},
  pages     = {308--318},
  year      = {2016},
  doi       = {10.1145/2976749.2978318},
}

@inproceedings{yao1986,
  author    = {Yao, Andrew Chi-Chih},
  title     = {How to Generate and Exchange Secrets},
  booktitle = {Proc. 27th Annual Symposium on Foundations of
               Computer Science (FOCS)},
  pages     = {162--167},
  year      = {1986},
  publisher = {IEEE},
  address   = {Toronto, ON, Canada},
  doi       = {10.1109/SFCS.1986.25},
}

@inproceedings{sirinam2018,
  author    = {Sirinam, Payap and Imani, Mohsen and Juarez, Marc
               and Wright, Matthew},
  title     = {Deep Fingerprinting: Undermining Website Fingerprinting
               Defenses with Deep Learning},
  booktitle = {Proc. ACM SIGSAC Conference on Computer and
               Communications Security (CCS)},
  pages     = {1928--1943},
  year      = {2018},
  doi       = {10.1145/3243734.3243768},
}

@inproceedings{wang2014,
  author    = {Wang, Tao and Cai, Xiang and Nithyanand, Rishab and
               Johnson, Rob and Goldberg, Ian},
  title     = {Effective Attacks and Provable Defenses for Website
               Fingerprinting},
  booktitle = {Proc. 23rd USENIX Security Symposium},
  pages     = {143--157},
  year      = {2014},
}

@techreport{pfitzmann2010,
  author      = {Pfitzmann, Andreas and Hansen, Marit},
  title       = {A Terminology for Talking about Privacy by Data
                 Minimization: Anonymity, Unlinkability, Undetectability,
                 Unobservability, Pseudonymity, and Identity Management},
  institution = {TU Dresden},
  year        = {2010},
  note        = {v0.34},
}

@techreport{3gpp_28552,
  author      = {{3GPP}},
  title       = {Management and Orchestration; 5G Performance Measurements},
  number      = {TS 28.552 V18.3.0},
  institution = {3rd Generation Partnership Project},
  year        = {2024},
}

@techreport{oran_wg1,
  author      = {{O-RAN Alliance}},
  title       = {O-RAN Architecture Description},
  number      = {O-RAN.WG1.O-RAN-Architecture-Description-v09.00},
  institution = {O-RAN Alliance},
  year        = {2023},
}

@misc{gdpr2016,
  author = {{European Parliament and Council}},
  title  = {Regulation ({EU}) 2016/679 (General Data Protection
            Regulation)},
  year   = {2016},
  note   = {Official Journal of the European Union, L~119},
}

@book{cover2006,
  author    = {Cover, Thomas M. and Thomas, Joy A.},
  title     = {Elements of Information Theory},
  edition   = {2nd},
  publisher = {Wiley-Interscience},
  address   = {Hoboken, NJ},
  year      = {2006},
  doi       = {10.1002/047174882X},
}

@techreport{3gpp_28312,
  author      = {{3GPP}},
  title       = {Management and Orchestration; Intent Driven Management Services for Mobile Networks},
  number      = {TS 28.312 V18.3.0},
  institution = {3rd Generation Partnership Project},
  year        = {2024},
}

@techreport{3gpp_28912,
  author      = {{3GPP}},
  title       = {Study on Enhanced Intent Driven Management Services for Mobile Networks},
  number      = {TR 28.912 V18.0.0},
  institution = {3rd Generation Partnership Project},
  year        = {2024},
}

@misc{eprivacy2002,
  author = {{European Parliament and Council}},
  title  = {Directive 2002/58/{EC} (ePrivacy Directive)},
  year   = {2002},
  note   = {Official Journal of the European Communities, L~201},
}

@misc{habib2025,
  author       = {Habib, Md Arafat and Elsayed, Medhat and Ozcan, Yigit
                  and Iturria-Rivera, Pedro Enrique and Bavand, Majid
                  and Erol-Kantarci, Melike},
  title        = {Generative {AI} for Intent-Driven Network Management
                  in {6G} {RAN}: A Case Study on the {Mamba} Model},
  year         = {2025},
  doi          = {10.48550/arXiv.2508.06616},
  note         = {arXiv preprint, \url{https://arxiv.org/abs/2508.06616}},
}

@article{wang2026,
  author  = {Wang, Yao and Yang, Chungang and Li, Tong and Ouyang, Ying
              and Mi, Xinru and Song, Yanbo},
  title   = {A Survey on Intent-Driven End-to-End {6G} Mobile Communication System},
  journal = {IEEE Communications Surveys \& Tutorials},
  volume  = {28},
  pages   = {882--915},
  year    = {2026},
  doi     = {10.1109/COMST.2025.3575041},
}

@article{leivadeas2023,
  author  = {Leivadeas, Aris and Falkner, Matthias},
  title   = {A Survey on Intent-Based Networking},
  journal = {IEEE Communications Surveys \& Tutorials},
  volume  = {25},
  number  = {1},
  pages   = {625--655},
  year    = {2023},
  doi     = {10.1109/COMST.2022.3215919},
}

@article{polese2023,
  author  = {Polese, Michele and Bonati, Leonardo and D'Oro, Salvatore and Basagni, Stefano and Melodia, Tommaso},
  title   = {Understanding {O-RAN}: Architecture, Interfaces, Algorithms, Security, and Research Challenges},
  journal = {IEEE Communications Surveys \& Tutorials},
  volume  = {25},
  number  = {2},
  pages   = {1376--1411},
  year    = {2023},
  doi     = {10.1109/COMST.2023.3239220},
}

@article{sharma2025,
  author  = {Sharma, Adit and Habibi Lashkari, Arash},
  title   = {A Survey on Encrypted Network Traffic: A Comprehensive
             Survey of Identification/Classification Techniques,
             Challenges, and Future Directions},
  journal = {Computer Networks},
  volume  = {257},
  pages   = {110984},
  year    = {2025},
  doi     = {10.1016/j.comnet.2024.110984},
}

@article{blika2024,
  author  = {Blika, Afroditi and Palmos, Stefanos and Doukas, George
             and Lamprou, Vangelis and Pelekis, Sotiris and Kontoulis,
             Michael and Ntanos, Christos and Askounis, Dimitris},
  title   = {Federated Learning For Enhanced Cybersecurity And
             Trustworthiness In {5G} and {6G} Networks: A Comprehensive
             Survey},
  journal = {IEEE Open Journal of the Communications Society},
  volume  = {6},
  pages   = {3094--3130},
  year    = {2025},
  doi     = {10.1109/OJCOMS.2024.3449563},
}

@techreport{tmforum_ig1305,
  author      = {{TM Forum}},
  title       = {Autonomous Networks: Empowering Digital Transformation
                 --- From Strategy to Implementation},
  number      = {IG1305},
  institution = {TM Forum},
  year        = {2023},
}

@techreport{etsi_zsm,
  author      = {{ETSI}},
  title       = {Zero-touch Network and Service Management ({ZSM}); Reference Architecture},
  number      = {GS ZSM 002 V1.1.1},
  institution = {ETSI},
  year        = {2019},
}

@article{blackwell1953,
  author  = {Blackwell, David},
  title   = {Equivalent Comparisons of Experiments},
  journal = {The Annals of Mathematical Statistics},
  volume  = {24},
  number  = {2},
  pages   = {265--272},
  year    = {1953},
  doi     = {10.1214/aoms/1177729032},
}

@article{montazerolghaem2026r21,
  author  = {Montazerolghaem, Ahmadreza},
  title   = {{R2+1}: Integrated Vehicle Routing, Packet Routing, and Resource
             Allocation for Knowledge-Defined {5G/6G} Multimedia Internet of
             Vehicles},
  journal = {Vehicular Communications},
  volume  = {60},
  pages   = {101054},
  year    = {2026},
  doi     = {10.1016/j.vehcom.2026.101054},
}

@article{montazerolghaem2026pid,
  author  = {Montazerolghaem, Ahmadreza},
  title   = {Energy-Efficient Software-Defined {5G/6G} Multimedia {IoV}:
             {PID} Controller-Based Approach},
  journal = {IEEE Transactions on Vehicular Technology},
  volume  = {75},
  number  = {6},
  pages   = {11621--11633},
  year    = {2026},
  doi     = {10.1109/TVT.2025.3649936},
}

\end{document}